\documentclass[aps,prx,groupedaddress,notitlepage,twocolumn,longbibliography,10pt]{revtex4-2}

\usepackage{amsmath,amssymb,amsthm,bm}
\usepackage{graphicx}
\usepackage[colorinlistoftodos]{todonotes}
\usepackage[inline]{enumitem}
\definecolor{darkblue}{rgb}{0.1,0.2,0.6}
\definecolor{darkred}{rgb}{0.8,0.1,0.2}
\definecolor{crimson}{RGB}{164,16,52}
\definecolor{darkgreen}{rgb}{0.31,0.62,0.24}
\usepackage[colorlinks=true, allcolors=crimson]{hyperref}
\usepackage{diagbox}
\usepackage{epigraph}
\usepackage{float}
\usepackage{multirow}
\usepackage{tikz}
\usepackage{mathtools}
\usepackage{braket}
\usepackage{soul} 
\usepackage[all]{xy}
\usepackage{lipsum}
\usepackage[export]{adjustbox}
\usepackage[thinlines]{easytable}

\usepackage{bbm}

\newcommand{\Rom}[1]{\uppercase\expandafter{\romannumeral#1}}

\newcommand{\mc}{\mathcal}

\DeclareMathOperator{\Tr}{Tr}

\DeclareMathOperator{\CZ}{CZ}

\makeatletter
\renewcommand*\env@matrix[1][*\c@MaxMatrixCols c]{%
	\hskip -\arraycolsep
	\let\@ifnextchar\new@ifnextchar
	\array{#1}}
\makeatother

\newtheorem*{claim*}{Claim}

\newtheorem*{proposition*}{Proposition}

\newcommand{\ScalingFactor}{0.8333}

\usepackage[normalem]{ulem}

\definecolor{darkred}{rgb}{0.8,0.1,0.2}

\newcommand*{\Edited}[1]{#1}

\begin{document}
\title{Constructing Bulk Topological Orders via Layered Gauging}
\author{Shang Liu}
\email{sliu.phys@gmail.com}
\affiliation{Institute of Physics,
Chinese Academy of Sciences, Beijing 100910, China}

\begin{abstract}
Understanding quantum phases and phase transitions in the presence of symmetries is a central objective of quantum many-body physics. A powerful modern paradigm for investigating this problem is topological holography, which relates symmetries in $k$ dimensions to ``bulk'' topological orders in $(k+1)$ dimensions. While conceptually profound, most existing bulk construction methods rely on sophisticated mathematical formalisms and can be difficult to apply to certain symmetry types. In this work, we propose a physically intuitive and versatile method, termed the layered gauging construction, to systematically generate $(k+1)$-dimensional (liquid or fracton) topological orders from $k$-dimensional generalized symmetries. Roughly speaking, the prescription is to stack many layers of $k$-dimensional quantum systems with certain symmetries into a $(k+1)$-dimensional pile, and then sequentially gauge a diagonal symmetry acting on each nearest-neighbor pair of layers. The detailed procedure depends on the specific symmetry types. We have successfully implemented the method in a number of examples in different spatial dimensions, with symmetries that are conventional, higher-form, subsystem, anomalous, nonabelian, or noninvertible. We hence conjecture the method to be very general. For example, from the subsystem symmetry of the $2d$ plaquette Ising model, we derive the X-cube model and also an anisotropic fracton topological order. Additionally, starting from an anomalous $\mathbb Z_2$ symmetry in $1d$, we construct a new square lattice model realizing the double semion topological order. From a quantum-error-correction perspective, our layered gauging construction may be viewed as a generalization of the hypergraph-product construction beyond the standard Calderbank-Shor-Steane setting.
\end{abstract}

\maketitle
\tableofcontents

\section{Introduction}
Understanding quantum phases and phase transitions in the presence of symmetries is a central objective of quantum many-body physics. In recent years, the framework of topological holography -- alternatively referred to as symmetry topological field theory or symmetry topological order -- has emerged as a powerful method for this purpose. See Ref.\,\onlinecite{Chen2025Essay} for a recent overview and Refs.\,\onlinecite{MooreSeiberg1989RCFT,Witten1989ChernSimons,Kong2015Center,Kong2018GaplessEdges,Kong201905MathGaplessI,Pulmann201909Sandwich,Ji201912NoninvAnomaly,Kong201912MathGaplessII,Berg202003TOBdry,Kong202003Classification,Kong202005Symmetry,Gaiotto202008Orbifold,Tachikawa202009SL2Z} for early developments. In this framework, associated to each (generalized) symmetry in $k$ dimensions, there is a corresponding $(k+1)$-dimensional bulk topological order. Many topological properties of $k$-dimensional quantum systems with the symmetry, such as the classification of gapped phases and the modular transformation rules of partition functions, can be captured by the algebraic properties of the anyon excitations in the bulk topological order. As a concrete example, by analyzing the anyon condensation patterns of the doubled-Ising topological order, one can show that in one spatial dimension and with the Kramers-Wannier duality as a symmetry, there is only one gapped phase with three-fold degenerate ground states \cite{Sakura20231dGapped}. As another example, in Ref.\,\onlinecite{Cheng2020RelativeAnomaly}, from the topological holographic viewpoint, the authors are able to prove that all unitary minimal conformal field theories in one spatial dimension have no 't Hooft anomaly, thereby imposing a strong constraint on gapless states with anomalous symmetries.  

Several concrete realizations of this holographic duality are already well-established. For instance, any fusion category symmetry in one spatial dimension ($1d$) corresponds to a $2d$ bulk topological order described by the Turaev-Viro topological quantum field theory (TQFT) \cite{TV1992StateSum,BW1996Invariants}, which admits a Hamiltonian realization through the Levin-Wen string-net models \cite{LevinWen2005}. As another example, a $0$-form finite group symmetry $G$ in arbitrary $k$ spatial dimensions, including those with 't Hooft anomalies, is associated with a Dijkgraaf-Witten TQFT in $(k+1)$ spatial dimensions \cite{DijkgraafWitten}. Hamiltonian realizations of these theories have been constructed for arbitrary anomalies at least in $2d$ \cite{KitaevTC,Wu2013TwistedQD}. 
As a common feature of most existing constructions, sophisticated mathematical tools, such as (higher) category theory, that describe the infrared properties of the systems are involved. Related to this feature, these constructions can not be applied to subsystem symmetries which are symmetries acting on rigid lower-dimensional subsystems. The bulk topological orders corresponding to subsystem symmetries are likely to be fracton topological orders, whose mathematical structure is less developed compared to conventional ``liquid'' topological orders. 

Recently, two new bulk construction methods following rather different ideas from the above were proposed \cite{LiuJi2023,GarreRubio2024IterGauging,GarreRubio2024StabilizerCodes,GarreRubio2025NonabelianIter}. The first method \cite{LiuJi2023} takes as input a pair of $k$-dimensional qubit lattice models related by a generalized Kramers-Wannier duality, and produces a $(k+1)$-dimensional topological order by a specific rule of constructing stabilizers. The second method \cite{GarreRubio2024IterGauging,GarreRubio2024StabilizerCodes,GarreRubio2025NonabelianIter}, named iterative gauging, starts by gauging a symmetry of a $k$-dimensional lattice model, followed by iteratively gauging the dual symmetry acting on the new gauge field degrees of freedom, until a $(k+1)$-dimensional lattice model is generated. 
These methods focus on the microscopic lattice Hamiltonians, rely less on the mathematical structures, and can both be applied to the case of subsystem symmetries. However, the two constructions also have their own limitations. The first method is only defined for onsite $\mathbb{Z}_2$ symmetries (although they can be higher-form or subsystem symmetries) and its generalization to nonabelian or non-onsite symmetries is not obvious. The second method is apparently incompatible with 't Hooft anomalous symmetries which can not be gauged. 

In this work, we propose a new bulk construction method dubbed \emph{layered gauging}, which may be regarded as a generalization of the two methods above but overcomes their aforementioned limitations. Roughly speaking, our prescription is to stack many layers of $k$-dimensional quantum systems with  certain symmetries into a $(k+1)$-dimensional pile, and then sequentially gauge a diagonal symmetry acting on each nearest neighbor pair of layers. The detailed procedure depends on the specific symmetry types. We have successfully implemented the method in a number of examples in different spacetime dimensions, with symmetries that are conventional, higher-form, subsystem, anomalous, nonabelian, or noninvertible. We hence conjecture the method to be very general, although a rigorous proof of its generality is lacking. For example, from the subsystem symmetry of the $2d$ plaquette Ising model, we constructed \emph{two} distinct fracton topological orders in $3d$ using two different ways of gauging. From the anomalous $\mathbb{Z}_2$ symmetry in $1d$, we constructed a new square lattice model realizing the double semion topological order. 

We hope our bulk construction method can produce more interesting lattice models in the future. The method may also become a basis for new quantum state preparation protocols. 

We shall mention two other related papers: (1) Our construction has some overlap with the layer construction of $3d$ topological orders in Ref.\,\onlinecite{Barkeshli2022HigherSymm}. In particular, the example in our Section \ref{sec:HigherFormExample} is operationally equivalent to the example in their Section 5.1.1. The difference is that Ref.\,\onlinecite{Barkeshli2022HigherSymm} mainly uses the language of anyon condensation and therefore relies more on the algebraic description of the anyon excitations in each layer. Also, it focuses on $3d$ liquid topological orders and does not touch on fracton models. (2) In a very recent paper Ref.\,\onlinecite{Tantivasadakarn2026HGP}, the authors reinterpreted product constructions of quantum error correcting (QEC) codes as a coupled-layer construction and discussed its relation to gauging. This overlaps with our examples when the bulk topological orders are Calderbank-Shor-Steane (CSS) codes, i.e those with all-$Z$ and all-$X$ stabilizers. We will discuss more about this in the Discussions section. 

The rest of the paper is organized as follows. In Section \ref{sec:PrescriptionSimplest}, we will introduce our bulk construction prescription and its intuition for the simplest symmetry types, followed by concrete examples in Sections \ref{sec:ElementaryExamples}-\ref{sec:SubsystemExamples}. 
Then in Section \ref{sec:PrescriptionGeneralized}, we generalize our method to more complex types of symmetries, followed by examples in Sections \ref{sec:AnomalyExample}-\ref{sec:NoninvertibleExample}. The examples presented in this paper are largely independent and need not be read in a sequential order. Finally in Section \ref{sec:Discussion}, we comment on possible future directions. 


\section{Prescription and Intuition in the Simplest Situation}\label{sec:PrescriptionSimplest}
\begin{figure}
    \centering
    \includegraphics[scale=\ScalingFactor]{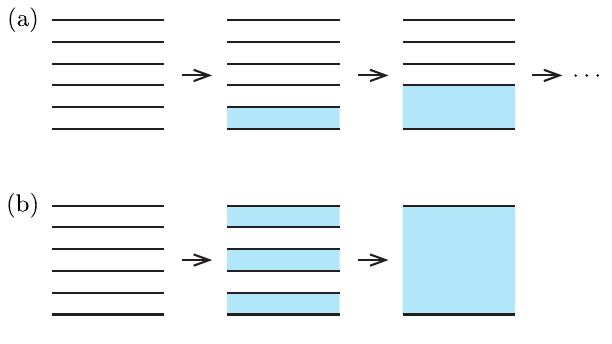}
    \caption{The layered gauging procedure. (a) Sequentially gauge a bilayer symmetry acting on every neighboring pair of layers $(n,n+1)$. From $n=1$ to $n=2$, $n=3$, and so on. This is the main approach taken in this paper. (b) Alternatively, we may take an even-odd gauging pattern: First gauge the bilayer symmetries acting on layer pairs $(2l-1,2l)$, and then gauge those acting on layer pairs $(2l,2l+1)$. }
    \label{fig:Prescription}
\end{figure}
In this section, we shall introduce our layered gauging prescription and its underlying intuition for the simplest symmetry types. Generalizations will be considered in Section \ref{sec:PrescriptionGeneralized}. 

Let $A$ be a symmetry acting on $k$-dimensional space, say some $k$-dimensional lattice for concreteness. As the simplest situation, suppose $A$ is \emph{invertible, abelian, and anomaly-free (onsite)}. For example, $A$ may be a conventional $\mathbb{Z}_2$ symmetry generated by the spin flip on all sites, and it may also be a higher-form or a subsystem symmetry. Denote the symmetry operators by $U_\alpha$ with some index $\alpha=1,2,\cdots$. The first step of our construction is to stack a series of $k$-dimensional quantum systems, all having the symmetry $A$, into a $(k+1)$-dimensional pile. This is illustrated by the left-most picture in Fig.\,\ref{fig:Prescription}a. The Hamiltonians of these layers will be specified later. 
Let us label the different $k$-dimensional layers by $n=1,2,3,\cdots$. The composite $(k+1)$-dimensional system now has an $A$ symmetry acting on every layer $n$, whose symmetry operators will be denoted by $U_{n,\alpha}$. Next, for every nearest-neighbor pair of layers $(n,n+1)$, we gauge the \emph{bilayer $A$ symmetry} generated by $U_{n,\alpha}U_{n+1,\alpha}^{-1}$ for all $\alpha$. This may be done in a sequential order as shown in Fig.\,\ref{fig:Prescription}a: First gauge the bilayer $A$ symmetry generated by $U_{1,\alpha}U_{2,\alpha}^{-1}$ for all $\alpha$, then gauge the symmetry generated by $U_{2,\alpha}U_{3,\alpha}^{-1}$ for all $\alpha$, then $U_{3,\alpha}U_{4,\alpha}^{-1}$, and so on.  

Suppose the $k$-dimensional layers are all boundary free, but we take open boundary condition along the layering direction. In other words, there is a bottommost layer $n=1$ and a topmost layer $n=N$. We claim that the $(k+1)$-dimensional bulk theory resulting from the above layered gauging procedure enforces an $A$ symmetry on its boundary. More precisely, $U_{1,\alpha}U_{N,\alpha}^{-1}=1$ for all $\alpha$. Indeed, after the gauging for the first two layers, the Gauss's law implies $U_{1,\alpha}U_{2,\alpha}^{-1}=1$. This is analogous to electromagnetism where the Gauss's law $\nabla\cdot E=\rho/\epsilon_0$ implies trivial total charge $Q=0$ on any closed space manifold \footnote{However, this property can not be directly generalized to the case of a nonabelian symmetry. }. Similarly, gauging the bilayer symmetry acting on the $n$-th and $(n+1)$-th layers enforces $U_{n,\alpha}U_{n+1,\alpha}^{-1}=1$. These altogether imply $U_{1,\alpha}U_{N,\alpha}^{-1}=1$ as claimed. 

\begin{figure}
    \centering
    \includegraphics[scale=\ScalingFactor]{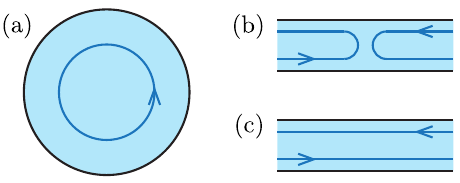}
    \caption{Bulk topological order and boundary symmetry. (a) A $2d$ topological order on a disk region and a circular anyon loop. (b) A cylinder region and a contractible anyon loop. (c) Antiparallel noncontractible anyon loops on a cylinder region. Panels (b) and (c) are equivalent if the anyon line $\mc N$ satisfies $\mc N\times \bar{\mc N}=1$. }
    \label{fig:BoundarySymmetryPicture}
\end{figure}
The property of an enforced boundary $A$ symmetry is exactly what we expect for a topological order characterizing the symmetry $A$ \cite{JiWen2019NoninvAnomaly}, i.e. a topological holographic dual theory. An intuitive picture is as follows. Say we have a $2d$ topological order (lowercase letter $d$ is for spatial dimension) on a disk region as shown in Fig.\,\ref{fig:BoundarySymmetryPicture}a. In the absence of bulk anyon excitations, all dynamical degrees of freedom of the system live on the boundary and hence can be described by an effective $1d$ theory. Now consider circular anyon loop operators as also illustrated in Fig.\,\ref{fig:BoundarySymmetryPicture}a. On one hand, when pushed to the boundary, these loops act as $1d$ symmetry operators and generate a boundary symmetry $A$. On the other hand, these loops are contractible in the bulk and thus evaluate to some positive numbers. This serves as a global constraint on the effective $1d$ theory and by which we say the $1d$ $A$ symmetry has been enforced. Note that in general, some of the anyon loop operators may act trivially in the effective $1d$ description of the boundary theory and thus nontrivial symmetry operators only correspond to a subset of the bulk anyon loops. As a concrete example, when the bulk topological order is the $2d$ $\mathbb{Z}_2$ toric code, the effective $1d$ theory lives on a spin chain subject to the global constraint $\prod_iX_i=1$, which comes from either the $e$ or $m$ anyon loop (but not both). This property of an enforced boundary symmetry can be generalized to a cylinder region which is relevant to our construction if we make an additional assumption: For every anyon line $\mc N$ corresponding to a nontrivial boundary symmetry operator, the fusion of $\mc N$ and its orientation reversed version $\bar {\mc N}$ is trivial, i.e. $\mc N\times \bar {\mc N}=1$ when properly normalized. See panels (b) and (c) of Fig.\,\ref{fig:BoundarySymmetryPicture}. Analogous arguments can also be carried out for higher dimensions. 

The above observation hints the possibility that the layered gauging construction leads to a topological holographic dual theory which we want. However, the $(k+1)$-dimensional bulk theory should ultimately depend on the Hamiltonians of the $k$-dimensional layers. Depending on our initial choice, the bulk may be a nontrivial topological order, but may also be trivially gapped (with a trivial boundary Hilbert space), or even be gapless. From the examples we examined, we conjecture that a good choice is to let the $k$-dimensional layers \emph{spontaneously break their $A$ symmetries completely}. If we regard each $k$-dimensional layer as an infinitesimally thin layer of $(k+1)$-dimensional topological order, the ground state degeneracy from spontaneous symmetry breaking can be regarded as its topological degeneracy. Gauging the bilayer symmetry on the first two layers grows this topological order to a thin slice of nonzero thickness. Gauging the bilayer symmetry on the next two layers further grows the topological order; see the shade in Fig.\,\ref{fig:Prescription}a. Eventually, it expands over a large $(k+1)$-dimensional space. 

\Edited{Interestingly, fully-symmetry-breaking initial Hamiltonian is not the only working choice. Following the suggestion of an anonymous referee, we considered an example where the initial Hamiltonian of each layer is in a symmetry protected topological (SPT) phase. The bulk model thus generated also has the right topological order corresponding to the symmetry. This example is presented in Appendix \ref{App:SPTExample}, which may be read after reading Section \ref{sec:2dTCExample}. We leave for future works to further investigate the effect of different initial Hamiltonians. }

It worth mentioning that we may alternatively implement the layered gauging in an even-odd pattern as shown in Fig.\,\ref{fig:Prescription}b. This might be useful for designing practical state preparation protocols, but in this paper, we will focus on the first, sequential way of gauging. 

Our detailed examples with the simplest type of symmetries (invertible, abelian, and anomaly-free) are given in Sections \ref{sec:ElementaryExamples}, \ref{sec:HigherFormExample}, and \ref{sec:SubsystemExamples}, where we respectively consider conventional zero-form symmetries, a higher-form symmetry, and a subsystem symmetry. 


\section{Elementary Examples}\label{sec:ElementaryExamples}
\subsection{2\textit{d} Toric Code from 1\textit{d} $\mathbb{Z}_2$ Ferromagnets}\label{sec:2dTCExample}
As the first example, let us try to construct a $2d$ topological order from $1d$ $\mathbb{Z}_2$ symmetric lattice models. Consider a 1$d$ periodic chain of qubits labeled by $j=1,2,\cdots,N_1$, and the Ising Hamiltonian 
\begin{align}
    H_{{\rm Ising},1d}=-J\sum_jZ_jZ_{j+1}-h\sum_jX_j, 
\end{align}
with $J,h\geq 0$. $H_{{\rm Ising},1d}$ has a $\mathbb{Z}_2$ symmetry generated by $U=\prod_jX_j$. Following our general prescription, we stack many copies of the 1$d$ qubit chain into a 2$d$ lattice as in Fig.\,\ref{fig:2dTCLatticeGrowth}a, where we label the different copies by the layer index $n=1,2,\cdots$. Each qubit is now labeled by a pair of indices $(n,j)$. The Hamiltonian has a decoupled form: 
\begin{align}
    H_0=\sum_{n,j}(-JZ_{n,j}Z_{n,j+1}-hX_{n,j}), 
\end{align}
which has a $\mathbb{Z}_2$ symmetry generated by $U_n:=\prod_j X_{n,j}$ on each layer $n$. Our goal is to gauge the subgroup symmetry generated by $U_n U_{n+1}$ (equivalent to $U_nU^{-1}_{n+1}$). 

\begin{figure*}
    \centering
    \includegraphics[scale=\ScalingFactor]{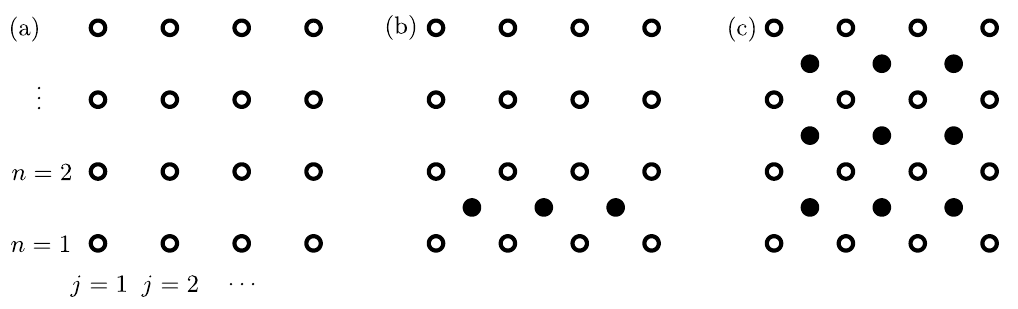}
    \caption{Lattice geometry evolution in the layered gauging construction of a $2d$ topological order. (a) Stacking multiple 1$d$ chains into a 2$d$ lattice. (b) Coupling the first two layers to gauge field degrees of freedom. (c) Final form of the bulk lattice. }
    \label{fig:2dTCLatticeGrowth}
\end{figure*}

Let us quickly review how gauging works using the Ising model $H_{{\rm Ising},1d}$. To gauge the onsite $\mathbb{Z}_2$ symmetry generated by $U$, we introduce new qubits with half-integer indices $j+1/2$, which represent the gauge field degrees of freedom. We then demand local gauge symmetries generated by $S_{j}:=X_{j-1/2}X_{j}X_{j+1/2}$ for all $j$. The $X$ terms in $H_{{\rm Ising},1d}$ already commute with all $S_j$ and hence need not be modified. In contrast, the $ZZ$ terms are not symmetric and we need to replace each $Z_jZ_{j+1}$ by $Z_jZ_{j+1/2}Z_{j+1}$, dubbed the minimal coupling procedure. The Hamiltonian now becomes
\begin{align}
    H_{{\rm Ising},1d}'=-J\sum_j Z_jZ_{j+1/2}Z_{j+1}-h\sum_j X_j. 
\end{align}
As the final step, we restrict the Hilbert space to the gauge invariant subspace, also called the physical subspace, where $S_j=1$ for all $j$. The condition $S_j=1$ is an analog of the Gauss's law in electromagnetism: $X_{j-1/2}X_jX_{j+1/2}\sim \exp[\pi i(E_{j+1/2}-E_{j-1/2}-Q_j)]$ where $E_{j\pm 1/2}\in \{0,1\}$ are the electric field and $Q_j\in \{0,1\}$ is the charge on site $j$. The Gauss's law constraints in particular imply that $U=\prod_j S_j=1$, i.e. the original $\mathbb{Z}_2$ symmetry has been enforced. For completeness, we have given a more general review of gauging in Appendix \ref{app:GaugingReview}, though it is not strictly necessary for reading this paper. 

Let us now get back to our layered gauging construction. As the first step, we would like to gauge the $\mathbb{Z}_2$ symmetry generated by $U_1 U_2$. We treat this symmetry as a 1$d$ symmetry and regard each rung of qubits, $Z_{1,j}$ and $Z_{2,j}$, as a single composite site. We therefore introduce gauge field qubits at $n=3/2$ and $j=1/2,3/2,\cdots, N_1-1/2$, as shown in Fig.\,\ref{fig:2dTCLatticeGrowth}b by the black dots. The local gauge symmetry generators, or Gauss's law stabilizers, take the following form: 
\begin{align}
    \includegraphics[valign=c,scale=\ScalingFactor]{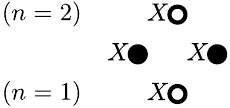}. 
    \label{eq:GaussLaw12TC}
\end{align}
The $ZZ$ operators in the first two layers of $H_0$ are modified by minimal coupling to  
\begin{align}
    \includegraphics[valign=c,scale=\ScalingFactor]{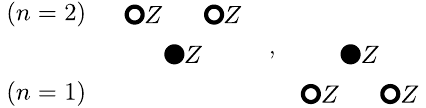}. 
    \label{eq:GaugedIsingTermsTC}
\end{align}
Single $X$ operators in $H_0$ commute with the gauge symmetry generators and hence need not be modified. 

Next, we proceed to gauge the $\mathbb{Z}_2$ symmetry generated by $U_2 U_3$. Note that this is still a valid symmetry after the previous gauging step. In a similar way, we introduce new gauge field sites between layers $2$ and $3$, and the new gauge symmetry generators take the same form as those in \eqref{eq:GaussLaw12TC} except for a shift in the layer indices. Single $X$ operators in the Hamiltonian and the old Gauss's law stabilizers from the previous step need not be modified. $Z$-type Hamiltonian terms (those that are products of $Z$'s) with support on the 2nd and 3rd layers become
\begin{align}
    \includegraphics[valign=c,scale=\ScalingFactor]{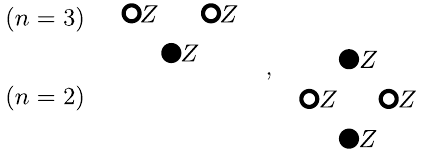}. 
\end{align}
Other $Z$-type terms are unaffected. 

It should be clear how to proceed further. We will eventually get a $2d$ model defined on a lattice that looks like Fig.\,\ref{fig:2dTCLatticeGrowth}c. If we take periodic boundary condition along the layering direction, i.e. $n\sim n+N_2$ for some integer $N_2$, the Hamiltonian we obtain is    
\begin{align}
    \tilde H= -h\left(\includegraphics[scale=\ScalingFactor,valign=c]{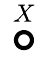}\right)
    - J\left(\includegraphics[scale=\ScalingFactor,valign=c]{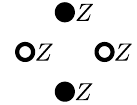}\right),  
\end{align}
where the sums over site indices are implicit. Moreover, the Hilbert space is subject to Gauss's law constraints of the form 
\begin{align}
    \includegraphics[scale=\ScalingFactor,valign=c]{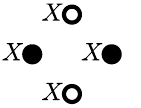}=1. 
\end{align}

The bulk model is particularly interesting when $h=0$ and $J>0$. In this case, it has the toric code topological order \cite{KitaevTC}, also called the $\mathbb{Z}_2$ topological order. To see this explicitly, we can substitute the white (black) dots in our bulk lattice by vertical (horizontal) links, then the four-$Z$ terms in $\tilde H$ become the plaquette stabilizers in the standard toric code, while the four-$X$ Gauss's law stabilizers become the star stabilizers. \Edited{It follows that $\tilde H$ shares the same ground states as the conventional toric code Hamiltonian, which is what we meant by the realization of the toric code topological order. However, $\tilde H$ has fewer possible excitations than the standard toric code model due to the hard Gauss's law constraints on the Hilbert space. To allow for richer excitations,} we may soften the Gauss's law constraints by adding the four-$X$ stabilizers to the Hamiltonian with a negative coefficients: 
\begin{align}
    H= &-h\left(\includegraphics[scale=\ScalingFactor,valign=c]{SingleMatterX.pdf}\right)
    - J\left(\includegraphics[scale=\ScalingFactor,valign=c]{ToricCodeZTerm.pdf}\right)\nonumber\\
    &-\lambda\left(\includegraphics[scale=\ScalingFactor,valign=c]{ToricCodeXTerm.pdf}\right). 
\end{align}
The topological phase of the system at $J>0,h=0$ is unaffected by this change as long as $\lambda>0$, since the Hamiltonian is frustration free. \Edited{Nonetheless, the softened Hamiltonian $H$ at $h=0$ becomes the conventional toric code model and allows for both $e$-anyon and $m$-anyon excitations, while the previous $\tilde H$ allows for only one type of anyons. }

If we take open boundary condition along the layering direction, and suppose $1\leq n\leq N_2$, the product of all Gauss's law stabilizers is equal to $U_1 U_{N_2}$. This suggests that our bulk model indeed enforces a boundary $\mathbb{Z}_2$ symmetry as expected. 

We may also think about what happens in the other limit $J=0$ and $h>0$. In this case, all white dots are pinned to the $\ket{+}$ state -- the eigenstate of $X$ with eigenvalue $+1$. The Gauss's law constraints then effectively reduce to $X_{n+1/2,j-1/2}X_{n+1/2,j+1/2}=1$, implying that each row of black dots spontaneously breaks the $\mathbb{Z}_2$ symmetry generated by $U_{n+1/2}:=\prod_j Z_{n+1/2,j+1/2}$. Note that the existence of these new $\mathbb{Z}_2$ symmetries is a consequence of the gauging procedure and is independent of the initial Hamiltonian. If we ignore the symmetries (allowing for arbitrary local perturbations), this bulk Hamiltonian does not belong to a well-defined phase, as the extensive ground state degeneracy can be lifted by arbitrarily weak perturbations. 

\subsection{3\textit{d} Toric Code from 2\textit{d} $\mathbb{Z}_2$ Ferromagnets}\label{sec:3dTCExample}
In the previous example, we have seen that the $2d$ toric code can be constructed from $1d$ $\mathbb{Z}_2$ ferromagnets using our layered gauging method. As the second example, we will show that the $3d$ toric code can be similarly constructed from $2d$ $\mathbb{Z}_2$ ferromagnets. 

Consider a 2$d$ square lattice of qubits labeled by integer pairs $(x,y)$ and the Ising Hamiltonian 
\begin{align}
    H_{{\rm Ising},2d}= &-J\sum_{x,y}(Z_{x,y}Z_{x+1,y}+Z_{x,y}Z_{x,y+1})\nonumber\\
    &-h\sum_{x,y}X_{x,y}
\end{align}
with $J,h\geq 0$. $H_{{\rm Ising},2d}$ has a $\mathbb{Z}_2$ symmetry generated by $U=\prod_{x,y}X_{x,y}$. Our goal is to stack together many layers of the above 2$d$ model, and then gauge the symmetries generated by $U_n U_{n+1}$, where as before, $U_n$ denotes the action of $U$ on the $n$-th layer.

Again, we start by reviewing the gauging procedure using the single-layer model $H_{{\rm Ising},2d}$ and its $\mathbb{Z}_2$ symmetry generated by $U$. We introduce one gauge field qubit between each pair of nearest-neighbor matter field sites. In other words, the gauge field qubits are in one-to-one correspondence with the links of the square lattice. We then demand local gauge symmetries generated by the following operators: 
\begin{align}
    \includegraphics[scale=\ScalingFactor,valign=c]{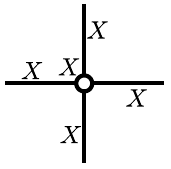}. 
    \label{eq:2dZ2GTGaussLaw}
\end{align}
Note that the product of all such operators is equal to the global symmetry generator $U$. Analogous to the 1$d$ Ising case, the single $X$ term in the Hamiltonian is already gauge invariant and need not be modified. In contrast, the $ZZ$ terms in the Hamiltonian need to couple to the gauge field degrees of freedom in order to enjoy the gauge symmetries. The simplest way of achieving this is to multiply each $ZZ$ by an additional $Z$ on the corresponding link. In addition, as a key difference between 1$d$ and 2$d$ gauge theories, now there exist some additional local operators, acting purely on gauge field sites, which commute with all existing Hamiltonian terms and the Gauss's law operators. These are called the flux operators and take the following form: 
\begin{align}
    \includegraphics[scale=\ScalingFactor,valign=c]{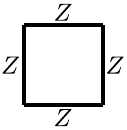}. 
\end{align}
Such an operator can be regarded as the analog of $e^{i\Phi}$ in electromagnetism where $\Phi\in\{0,\pi\}$ is the magnetic flux. As a customary step of the gauging procedure, we also add the flux operators to the Hamiltonian with a negative coefficient, in analogy with the magnetic part of Maxwell term. Putting all those ingredients together, the Hamiltonian after gauging becomes
\begin{align}
    H_{{\rm Ising},2d}'=& - J\left(
    \includegraphics[scale=\ScalingFactor,valign=c]{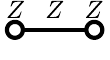}+\includegraphics[scale=\ScalingFactor,valign=c]{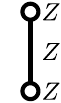}
    \right)-h\left(\includegraphics[scale=\ScalingFactor,valign=c]{SingleMatterX.pdf}\right)\nonumber\\
    &-\gamma\left(\includegraphics[scale=\ScalingFactor,valign=c]{2dZ2GTFluxTerm.pdf}\right). 
    \label{eq:Gauged2dIsing}
\end{align}
Finally, we restrict the Hilbert space to the gauge invariant subspace where the stabilizers in \eqref{eq:2dZ2GTGaussLaw} all equal to $1$. Note that if we did not add the flux operators, the gauge theory would be pathological in the $J\rightarrow 0$ limit with a large ground state degeneracy and extensive massless excitations. 

Getting back to our layered gauging construction, we can now apply  similar gauging procedures to the symmetries generated by $U_{n} U_{n+1}$ sequentially. When the $n$-th and $(n+1)$-th layers are considered, we can treat each rung of qubits labeled by $(n,x,y)$ and $(n+1,x,y)$ as a single composite site, as we did in the previous $2d$ toric code example. 
We will skip the detailed gauging steps and directly present the final result. It will be convenient to think of the bulk lattice as the 3$d$ cubic lattice with qubits living on the \emph{links}. The $z$-direction is the layering direction, vertical links ($z$-links) are the matter field degrees of freedom coming from the original 2$d$ models, and horizontal links ($x$-, $y$-links) are the gauge field degrees of freedom introduced during the layered gauging procedures. Setting $h=0$, corresponding to the symmetry breaking phase of $H_{{\rm Ising},2d}$, and after softening the Gauss's law constraints, the $3d$ Hamiltonian takes the following form:  
\begin{align}
    H=& -J\left( \includegraphics[scale=\ScalingFactor,valign=c]{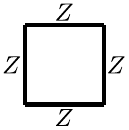}+\includegraphics[scale=\ScalingFactor,valign=c]{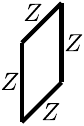} \right)-\gamma\left( \includegraphics[scale=\ScalingFactor,valign=c]{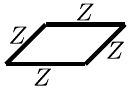} \right)\nonumber\\
    &-\lambda\left( \includegraphics[scale=\ScalingFactor,valign=c]{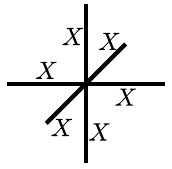} \right).\hspace{5em} \includegraphics[scale=\ScalingFactor,valign=c]{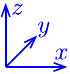} 
\end{align}
This realizes the $3d$ toric code or $\mathbb{Z}_2$ topological order. \Edited{We note that terms in this Hamiltonian have three different origins: 
\begin{enumerate}
    \item \textit{Those inherited from the lower-dimensional Hamiltonians}, which has the coefficient $-J$ in the above example. 
    \item \textit{Flux terms introduced by the gauging}, which has the coefficient $-\gamma$. 
    \item \textit{Gauss's law terms}, which has the coefficient $-\lambda$. 
\end{enumerate}
This is also the general structure of bulk Hamiltonians obtained by the layered gauging method with softened Gauss's law constraints.} 

In a similar way, toric code with star and plaquette stabilizers in arbitrary $(k+1)$ dimensions can be generated from $k$-dimensional $\mathbb{Z}_2$ ferromagnets. 

\section{An Example with Higher-Form Symmetry}\label{sec:HigherFormExample}
The $3d$ toric code topological order can also be constructed by layered gauging a one-form $\mathbb{Z}_2$ symmetry, as we will explain in this section. 

Consider a 2$d$ square lattice with qubits living on the \emph{links} and the following Hamiltonian: 
\begin{align}
    H_{Z_2{\rm GT}}=& - J\left(
    \includegraphics[scale=\ScalingFactor,valign=c]{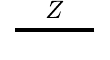}+\includegraphics[scale=\ScalingFactor,valign=c]{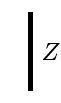}
    \right)\nonumber\\
    &-h\left(\includegraphics[scale=\ScalingFactor,valign=c]{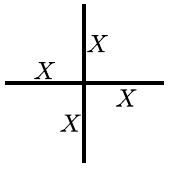}\right)-\gamma\left(\includegraphics[scale=\ScalingFactor,valign=c]{2dZ2GTFluxTerm.pdf}\right). 
\end{align}
This model is actually \emph{equivalent} to the gauged Ising model in \eqref{eq:Gauged2dIsing}: The Gauss's law of the gauged Ising model indicates that the site (matter field) degrees of freedom are completely redundant, as the value of $X$ on every site can be inferred from the values of $X$ on the surrounding links. Hence, we can use the Gauss's law constraints to remove all site qubits \footnote{An explicit way of achieving this is to apply ${\rm CNOT}$ gates to every pair $(v,e)$, where $v$ is a site, $e$ is a link adjacent to $v$, and the site qubit is the controlling qubit. As a result, the Gauss's law becomes $X_v=1$ for every site $v$. Therefore, site qubits are all frozen and can be effectively removed. }, and the resulting Hamiltonian is exactly $H_{Z_2{\rm GT}}$ above. 

The Hamiltonian $H_{Z_2{\rm GT}}$ has a \emph{one-form} symmetry whose generators act on loops or more generally codimension-one submanifolds. Explicitly, given any loop $l$ formed by a closed chain of links $e\in l$, $U_l:=\prod_{e\in l}Z_e$ commutes with the Hamiltonian. We would like to apply our layered gauging method to this one-form symmetry: Stack many copies of the $2d$ gauge theory $H_{Z_2{\rm GT}}$ and then gauge the symmetries generated by $U_{n,l}U_{n+1,l}$ for all $n$ and $l$. Here, as before, the layer index $n$ labels the different copies of the 2$d$ theory and $U_{n,l}$ is the action of $U_l$ on the $n$-th copy. 

Again, we will first practice gauging using the single-layer model $H_{Z_2{\rm GT}}$, and the bilayer cases are not much different. As $U_l$ acts on a loop, it looks very much like a 1$d$ symmetry generator. Comparing with our previous 1$d$ gauging example, it is natural to guess that we should introduce gauge field qubits on the \emph{sites} of the square lattice and declare gauge symmetries generated by the following operators. 
\begin{align}
    \includegraphics[scale=\ScalingFactor,valign=c]{2dZ2GTxHopping.pdf},\quad
    \includegraphics[scale=\ScalingFactor,valign=c]{2dZ2GTyHopping.pdf}. 
    \label{eq:1FormGaugingGaussLaw}
\end{align}
The four-$X$ terms in the Hamiltonian can be minimally modified to the following form in order to commute with the above generators: 
\begin{align}
    \includegraphics[scale=\ScalingFactor,valign=c]{2dZ2GTGaussLaw.pdf}. 
\end{align}
On the other hand, $Z$-type operators in the Hamiltonian need not be modified. We therefore obtain the Hamiltonian
\begin{align}
    H_{Z_2{\rm GT}}'=& - J\left(
    \includegraphics[scale=\ScalingFactor,valign=c]{2dZ2GTxHoppingMatterInt.pdf}+\includegraphics[scale=\ScalingFactor,valign=c]{2dZ2GTyHoppingMatterInt.pdf}
    \right)\nonumber\\
    &-h\left(\includegraphics[scale=\ScalingFactor,valign=c]{2dZ2GTGaussLaw.pdf}\right)-\gamma\left(\includegraphics[scale=\ScalingFactor,valign=c]{2dZ2GTFluxTerm.pdf}\right),  
\end{align}
together with the Hilbert space constraint that operators in \eqref{eq:1FormGaugingGaussLaw} all equal to $1$. This new theory, resulting from the gauging of a one-form symmetry, is actually equivalent to the $2d$ Ising model $H_{{\rm Ising}, 2d}$ presented in Section \ref{sec:3dTCExample}. This is because all \emph{link} degrees of freedom in this theory are redundant and can be removed by the Gauss's law \footnote{This will be more evident if we first apply CNOT gates to every site-link pair $(v,e)$. }. After this transformation, $H_{Z_2{\rm GT}}$ becomes the same as $H_{{\rm Ising}, 2d}$ up to an additive constant coming from the $\gamma$ terms. The one-form symmetry in $H_{Z_2{\rm GT}}$ is sometimes called the dual symmetry of the conventional $\mathbb{Z}_2$ symmetry in $H_{{\rm Ising}, 2d}$. 

Getting back to our layered gauging construction, bilayer one-form symmetries generated by $U_{n,l}U_{n+1,l}$ can be gauged in a similar way. 
Let the $H_{Z_2{\rm GT}}$ layers be parallel to the $xy$-plane and let the new gauge field qubits be represented by $z$-links. After setting $J=0$, that is when the one-form symmetry in each layer is spontaneously broken, and after softening the Gauss's law constraints, the $3d$ bulk Hamiltonian is  
\begin{align}
    H=& -\lambda\left( \includegraphics[scale=\ScalingFactor,valign=c]{3dTCZTermxz.pdf}+\includegraphics[scale=\ScalingFactor,valign=c]{3dTCZTermyz.pdf} \right)-\gamma\left( \includegraphics[scale=\ScalingFactor,valign=c]{3dTCZTermxy.pdf} \right)\nonumber\\
    &-h\left( \includegraphics[scale=\ScalingFactor,valign=c]{3dTCXTerm.pdf} \right).\hspace{5em} \includegraphics[scale=\ScalingFactor,valign=c]{3DFrame.pdf} 
\end{align}
This again realizes the $3d$ toric code topological order. 

\Edited{As already noted in the introduction section, this one-form symmetry example is operationally equivalent to the coupled layer construction in Section 5.1.1 of Ref.\,\onlinecite{Barkeshli2022HigherSymm}. It is worth mentioning that there is another interesting example in Section 5.1.2 of the same reference, which may be regarded as a layered gauging construction of the $3d$ toric code with a fermionic charge (also called the fermionic toric code). The single-layer Hamiltonian is again $H_{Z_2{\rm GT}}$ with $J=0$, nothing but the $2d$ toric code model. The relevant single-layer symmetry is the one-form symmetry generated by the fermionic loop. There is, however, a subtlety in this construction as already noted in Ref.\,\onlinecite{Barkeshli2022HigherSymm}: The Gauss's law stabilizers contain fermion hopping operators as ingredients, but there are multiple ways of writing down fermion hopping operators (small segments of fermion loops) and only certain choices of them lead to compatible Gauss's law constraints. This is related to the fact that the fermion loops are action ribbons of a nonzero width. It is unclear to us whether this step can be made more natural. }

\section{Examples with Subsystem Symmetry}\label{sec:SubsystemExamples}
\subsection{The 2\textit{d} Plaquette Ising Symmetry}
The $2d$ Plaquette Ising model is one of the simplest lattice models with subsystem symmetries. Consider a 2$d$ square lattice of qubits labeled by integer pairs $(x,y)$, the plaquette Ising Hamiltonian reads 
\begin{align}
    H_\text{plaq-Ising}=& -J\sum_{x,y}Z_{x,y}Z_{x+1,y}Z_{x,y+1}Z_{x+1,y+1}\nonumber\\
    & -h\sum_{x,y}X_{x,y}. 
\end{align}
This model has $\mathbb{Z}_2$ symmetries acting on 1$d$ subsystems: The symmetry generators are $U_x:=\prod_y X_{x,y}$ and $U_y:=\prod_x X_{x,y}$ supported on any column or row of the square lattice. Note the difference between subsystem symmetry considered here and higher-form symmetry considered in Section \ref{sec:HigherFormExample}: Subsystem symmetry generators act on \emph{rigid subsystems}, such as straight lines in the present case, while higher-form symmetry generators can act on arbitrary submanifolds of some fixed codimension, such as all 1$d$ lines, either straight or curvy. 

Next, we would like to stack many copies of the $2d$ plaquette Ising model labeled by the layer index $n$, and then gauge the $\mathbb{Z}_2$ symmetries generated by $U_{n,x}U_{n+1,x}$ and $U_{n,y}U_{n+1,y}$. As it turns out, there are at least two different ways of gauging such subsystem symmetries and they lead to inequivalent $3d$ fracton topological orders. In what follows, we will first explain the two ways of gauging using the single-layer model $H_\text{plaq-Ising}$ for its simplicity, and then present the bulk construction results from these two different routes. 

\subsection{Two Ways of Gauging}
One obvious way of gauging the subsystem symmetry in $H_\text{plaq-Ising}$ is to sequentially gauge the $\mathbb{Z}_2$ symmetries on each row and column, treating them as 1$d$ symmetries. The ordering of the gauging sequence does not matter and the resulting Hamiltonian is as follows. 
\begin{align}
    H_\text{plaq-Ising}^{(1)}=
    -J\left(\includegraphics[scale=\ScalingFactor,valign=c]{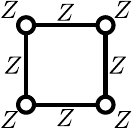}\right)-h\left(\includegraphics[scale=\ScalingFactor,valign=c]{SingleMatterX.pdf}\right), 
\end{align}
where qubits on the links are gauge field degrees of freedom. In addition, the Hilbert space is subject to the following Gauss's law: 
\begin{align}
    \includegraphics[scale=\ScalingFactor,valign=c]{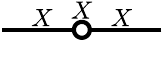},\quad \includegraphics[scale=\ScalingFactor,valign=c]{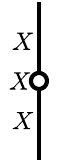}\quad=1. 
\end{align}
We would like to note a possible simplification of the above model, although this will not be used for our later bulk construction: We may remove the site degrees of freedom using half of the Gauss's law constraints, analogous to what we did in Section \ref{sec:HigherFormExample}. The resulting equivalent model has the following Hamiltonian 
\begin{align}
    H_\text{dual-PI}^{(1)}= -J
    \left( \includegraphics[scale=\ScalingFactor,valign=c]{2dZ2GTFluxTerm.pdf} \right)-h\left( \includegraphics[scale=\ScalingFactor,valign=c]{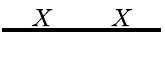} \right), 
\end{align}
and is subject to the constraints
\begin{align}
    \includegraphics[scale=\ScalingFactor,valign=c]{2dZ2GTGaussLawMatterInt.pdf}=1. 
\end{align}
This simplified model was first discovered in Ref.\,\onlinecite{LiuJi2023} as an unconventional Kramers-Wannier dual to the plaquette Ising model. It is not hard to check that the Hamiltonian terms of $H_\text{plaq-Ising}$ and $H_\text{dual-PI}^{(1)}$ satisfy the same commutation or anticommutation relations. 

There exists a different way of gauging the subsystem symmetry in $H_\text{plaq-Ising}$. In this second approach, we introduce gauge field qubits at \emph{plaquette centers} and 
declare the following gauge symmetry generators: 
\begin{align}
    \includegraphics[scale=\ScalingFactor,valign=c]{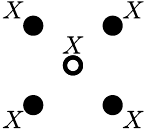}, 
\end{align}
where the black dots are gauge field qubits. Note that the subsystem symmetry generators $U_x$ and $U_y$ can all be decomposed as the products of such local operators. Therefore, the subsystem symmetry will be enforced once we demand the above operators all equal to $+1$ (Gauss's law). To commute with those gauge symmetry generators, each four-$Z$ term in $H_\text{plaq-Ising}$ should be multiplied by another $Z$ at the corresponding plaquette center. The gauged Hamiltonian is then
\begin{align}
    H_\text{plaq-Ising}^{(2)}=
    -J\left(\includegraphics[scale=\ScalingFactor,valign=c]{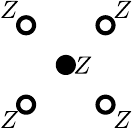}\right)-h\left(\includegraphics[scale=\ScalingFactor,valign=c]{SingleMatterX.pdf}\right).  
\end{align}
We may use the Gauss's law constraints to remove all matter field degrees of freedom, leading to the simplified model given below. 
\begin{align}
    H_\text{dual-PI}^{(2)}=
    -J\left(\includegraphics[scale=\ScalingFactor,valign=c]{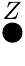}\right)-h\left(\includegraphics[scale=\ScalingFactor,valign=c]{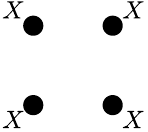}\right).  
\end{align}
This is nothing but the $2d$ plaquette Ising model with the roles of $X$ and $Z$ exchanged. $H_\text{dual-PI}^{(2)}$ is the more conventional Kramers-Wannier dual to $H_\text{plaq-Ising}$. This second way of gauging may be understood as follows: We regard each plaquette as a ``generalized link'' connected to four sites, and then each four-$Z$ term in $H_\text{plaq-Ising}$ can be regarded as a generalized $1d$ Ising term. The model can then be gauged in a similar way as the $1d$ Ising model. This perspective is more natural from the theory of classical error correcting codes \cite{Vedika2023LDPC1}, where we regard each four-$Z$ term as a classical check. 

It might be interesting to note that $H_\text{dual-PI}^{(2)}$ can be mapped to $H_\text{dual-PI}^{(1)}$ by a standard Kramers-Wannier duality, which is equivalent to gauging the $\mathbb{Z}_2$ symmetry generated by the product of all $Z$, removing matter sites by the Gauss's law, and finally setting the flux term coefficient to $+\infty$. 

\subsection{X-Cube Fracton Topological Order}
Both ways of gauging the subsystem $\mathbb{Z}_2$ symmetries can be generalized to the bilayer case and hence can be adopted for our layered gauging construction. Interestingly, these two approaches lead to distinct bulk topological orders, as we will see in this and the next subsection. We note that topological holography in the case of subsystem symmetry is still an underexplored subject \cite{Bulmash2019,Luo2022XCube,LiuJi2023,Fontana2022ChamonModelBdry,Cao2023FractonSymTFT,Nat2023FractonBdry}. Nonetheless, previous works have found that in this context, the same boundary theory may correspond to distinct bulk fracton topological orders \cite{LiuJi2023,Nat2023FractonBdry}.  

According to previous experience, we set $h=0$ and $J>0$ so that the subsystem symmetry of $H_\text{plaq-Ising}$ is spontaneously broken. 
If we take the first way of gauging, that is to gauge the 1$d$ symmetry generators $U_{n,x}U_{n+1,x}$ and $U_{n,y}U_{n+1,y}$ sequentially, the resulting bulk Hamiltonian after softening the Gauss's law is as follows. 
\begin{align}
    H^{(1)}=&-\lambda\left(\includegraphics[scale=\ScalingFactor,valign=c]{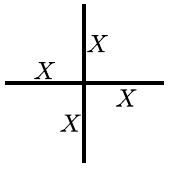}+\includegraphics[scale=\ScalingFactor,valign=c]{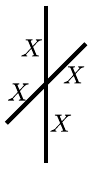}\right)\nonumber\\
    &-J\left(\includegraphics[scale=\ScalingFactor,valign=c]{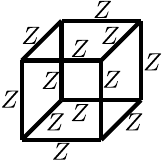}\right).\hspace{4em} \includegraphics[scale=\ScalingFactor,valign=c]{3DFrame.pdf}
    \label{eq:XCubeModel}
\end{align}
As the pictures suggest, the bulk model qubits live on the links of a 3$d$ cubic lattice. $z$-links come from the 2$d$ plaquette Ising layers, while $x$- and $y$-links support gauge field degrees of freedom introduced by the layered gauging. The $3d$ model realizes the renowned X-cube fracton topological order \cite{XCubeModel}. The standard X-cube model has another type of four-body stabilizers that are parallel to the $xy$-plane. However, those stabilizers are not independent from the existing ones in the Hamiltonian above. Therefore, missing those terms will not affect the topological order of the ground state. 

The X-cube fracton topological order has a robust and system size dependent ground state degeneracy. Suppose the cubic lattice is periodic with $L_x\times L_y\times L_z$ number of vertices, this degeneracy $D$ is given by $\log_2 D=2(L_x+L_y+L_z)-3$. In contrast, the ground state degeneracy of a conventional or ``liquid'' topological order, such as the $3d$ toric code, only depends on the topology of the space manifold. As a related property, anyon excitations in the X-cube fracton topological order has restricted mobility. 

\subsection{An Anisotropic Fracton Model}
Alternatively, we may take the second way of gauging, that is to introduce gauge field qubits at centers of plaquettes (or plaquette pairs in the bilayer case). The resulting 3$d$ bulk lattice may be decomposed into two sublattices: Qubits in the first sublattice are located at integer coordinates $(x,y,z)\in \mathbb{Z}^3$ and are represented by white dots. Qubits in the second sublattice are located at half integer coordinates $(x+1/2,y+1/2,z+1/2)$ and are represented by black dots. 
After setting $h=0$ and softening the Gauss's law, the bulk Hamiltonian is 
\begin{align}
    H^{(2)}= -J\left(\includegraphics[scale=\ScalingFactor,valign=c]{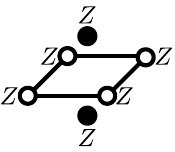}\right)-\lambda\left(\includegraphics[scale=\ScalingFactor,valign=c]{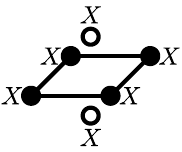}\right) 
\end{align}
where the 3$d$ coordinate frame is the same as that in \eqref{eq:XCubeModel}. This model was first described in Ref.\,\onlinecite{Chen2019Gauging} and was subsequently also constructed by a few other methods \cite{You2018MajoranaFracton,Fuji2019CoupledLayer,LiuJi2023,Zhai2025Product}. It realizes an anisotropic fracton topological order: The ground state degeneracy depends on the system sizes along the $x$ and $y$ directions but not on that along the $z$ direction. Moreover, the anyon excitations have restricted mobility along the $x$ and $y$ directions but are free to move along the $z$ direction. Therefore, $H^{(2)}$ and $H^{(1)}$ are topologically inequivalent fracton models. 

\section{Generalized Prescriptions}\label{sec:PrescriptionGeneralized}
With suitable modifications, our layered gauging construction can also be applied to more general types of symmetries. 

In Section \ref{sec:AnomalyExample}, we consider a symmetry $A$ that is invertible, abelian, but has \emph{'t Hooft anomaly} (obstruction to gauging). Although the $A$ symmetry acting on each single layer can not be gauged due to the anomaly, the bilayer symmetry generated by $U_{n,\alpha}U_{n+1,\alpha}^{-1}$ is anomaly free and can still be gauged. 
The key distinction from the case of onsite symmetries is as follows: After gauging the bilayer symmetry on the first two layers, $U_{2,\alpha}$ do not commute with the Gauss's law operators (gauge symmetry generators) and are no longer valid symmetry operators. This issue can be resolved by coupling the second-layer symmetry operators to the gauge field in a way that respects the Gauss's law. Let the modified symmetry operators from this procedure be $U_{2,\alpha}'$. We can now proceed to gauge $U_{2,\alpha}'U^{-1}_{3,\alpha}$, then modify $U_{3,\alpha}$ to $U_{3,\alpha}'$, gauge $U_{3,\alpha}'U_{4,\alpha}^{-1}$, and so on. 

In Section \ref{sec:NonabelianExample}, we consider a symmetry $G$ that is invertible, anomaly-free, but \emph{nonabelian}. 
If we still denote the symmetry operators by $U_{n,\alpha}$, an immediate issue is that $U_{n,\alpha}U^{-1}_{n+1,\alpha}$ do not generate the same symmetry group as $G$. This is not the biggest issue as we may instead try to gauge $U_{n,\alpha}U_{n+1,\alpha}$ \footnote{Indeed, in the case of an abelian invertible symmetry, under group redefinitions (automorphisms) on half of the layers, we may also state our construction as gauging the bilayer symmetries generated by $U_{n,\alpha}U_{n+1,\alpha}$. }. The real problem is that: Since the symmetry group $G$ itself is nonabelian, there is no way to expect $U_{2,\alpha}$ to respect the Gauss's law after gauging the bilayer symmetry on the first two layers. Coupling $U_{2,\alpha}$ to the gauge fields does not seem to help. The way out we find is as follows. We require each layer to have both a ``left'' symmetry $G_L$ and a ``right'' symmetry $G_R$ that mutually commute, although the representation of $G_L\times G_R$ need not be faithful. Let the action of $G_L$ ($G_R$) on the $n$-th layer be $\rho_{n,L}(g)$ ($\rho_{n,R}(g)$) with $g\in G$. Our generalized construction is to sequentially gauge the bilayer $G$ symmetry generated by $\rho_{n,L}(g)\rho_{n+1,R}(g)$. Apparently, $G_L$ and $G_R$ on each layer need to be somehow ``glued'' together, as the construction will certainly not work if each layer decomposes into two standalone subsystems supporting two independent $G$ symmetries. We achieve this by requiring that each layer initially fully break both $G_L$ and $G_R$, but preserve a $G$ subgroup of $G_L\times G_R$. In the special case where $G$ is abelian, we can choose $\rho_{n,L}(g)=\rho_{n,R}(g)^{-1}=U_{n,g}$ such that the diagonal subgroup of $G_L\times G_R$ acts trivially and is unbroken. The generalized construction then reduces to its simplest version. 

Finally in Section \ref{sec:NoninvertibleExample}, we consider an example with a \emph{noninvertible symmetry} (fusion category symmetry) $\mc C$ that is commutative and free of 't Hooft anomaly (but not onsite). Denote the symmetry generators by $\mc N_\mu$ with some index $\mu$. Associated to each $\mc N_\mu$, there is a dual symmetry generator $\bar{\mc N}_\mu$ satisfying $\mc N_\mu\times \bar{\mc N}_\mu=1+\cdots$. A natural generalization of our previous construction is to sequentially gauge the bilayer symmetries generated by $\mc N_{n,\mu}\bar{\mc N}_{n+1,\mu}$. However, these operators are not closed under multiplication and do not generate a $\mc C$ symmetry, which is because the product of two symmetry generators in a noninvertible symmetry is generically a sum of several distinct symmetry generators. Nonetheless, based on the matrix product operator structure of $\mc N_\mu$, we have found a natural gauging-like procedure to promote the bilayer global symmetry operators $\mc N_{n,\mu}\bar{\mc N}_{n+1,\mu}$ to local ones. Moreover, the resulting local ``gauge'' symmetry operators do generate a $\mc C$ symmetry. Our bulk construction works out nicely with this ``generalized gauging'' at least in the example we studied. 


\section{An Example with Anomaly}\label{sec:AnomalyExample}
\subsection{Anomalous $\mathbb{Z}_2$ Symmetry in 1\textit{d}}
Consider a 1$d$ periodic chain of qubits labeled by $j=1,2,\cdots,N_1$. The following unitary generates an anomalous $\mathbb{Z}_2$ symmetry \cite{LevinGu}: 
\begin{align}
    U:=-\prod_j X_j \prod_k L_{k,k+1},
    \label{eq:1dAnomalousZ2Def}
\end{align}
where
\begin{align}
    L_{j,k}:=i^{(1-Z_jZ_k)/2}. 
    \label{eq:LjkDef}
\end{align}
The effect of $L_{k,k+1}$ is to assign a factor of $i$ to each domain wall. The overall $-1$ factor in $U$ is unimportant and is just to match with the convention in Ref.\,\onlinecite{LevinGu}. This anomalous $\mathbb{Z}_2$ symmetry was first derived in Ref.\,\onlinecite{LevinGu} as the boundary symmetry of a symmetry protected topological state. One may also directly compute its anomaly data using, e.g. the Else-Nayak approach \cite{ElseNayak}. As a consequence of the nontrivial anomaly, any local Hamiltonian with this symmetry can not have a trivially gapped ground state: Either there is spontaneous symmetry breaking (SSB) or the spectrum is gapless. 

The symmetry generator $U$ can be rewritten in a different form that will be useful for our later calculations. We use the notation $\Lambda_j(M)$ to denote controlled unitary gates: 
\begin{align}
    \Lambda_j(M):=\left(\frac{1+Z_j}{2}\right)+M\left(\frac{1-Z_j}{2}\right), 
\end{align}
where $M$ is a unitary operator acting on some sites other than the controlling site $j$. The single-qubit phase gate $S_j:=\Lambda_j(i)$ (equal to ${\rm diag}(1,i)$ in the $Z$-basis) and the two-qubit controlled-$Z$ gate $\CZ_{k,j}=\CZ_{j,k}:=\Lambda_j(Z_k)$ are frequently used examples in the context of quantum circuits. One can verify that 
\begin{align}
    L_{j,k}=S_jS_k\CZ_{j,k}.  
\end{align}
It follows that 
\begin{align}
    U=-(-i)^{N_1}\prod_j Y_j\prod_k\CZ_{k,k+1}. 
\end{align}

\subsection{Bulk Construction}
\begin{figure}
    \centering
    \includegraphics[scale=\ScalingFactor]{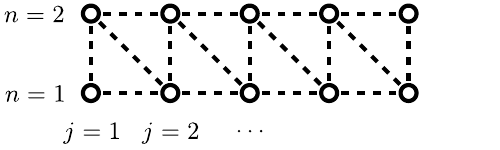}
    \caption{Triangular plaquettes for the purpose of layered gauging $1d$ anomalous $\mathbb{Z}_2$ symmetry. }
    \label{fig:TriangularPlaquettes}
\end{figure}
We would like to construct a $2d$ topological order that enforces this anomalous symmetry on its boundary. To this end, we stack many copies of the 1$d$ chain to form a 2$d$ lattice as in Fig.\,\ref{fig:2dTCLatticeGrowth}a, with the different copies labeled by the layer index $n=1,2,\cdots$. 
Sites or qubits in this lattice are then labeled by a pair of integers $(n,j)$. We choose the initial decoupled Hamiltonian to be 
\begin{align}
    H_0=-\sum_{n,j}Z_{n,j}Z_{n,j+1}. 
\end{align}
For each layer $n$, we can define a 1$d$ anomalous $\mathbb{Z}_2$ generator $U_n$ as in Eq.\,\ref{eq:1dAnomalousZ2Def}. Under $H_0$ and these symmetries, all layers are in the SSB phase. Our goal now is to sequentially gauge the $\mathbb{Z}_2$ subsystem symmetries generated by $U_n U_{n+1}$, starting from $n=1$. Importantly, the $\mathbb{Z}_2$ symmetry generated by $U_n U_{n+1}$ is anomaly free and hence the gauging is possible. 
We will see that the resulting theory realizes the double semion topological order. 

\textbf{Step 1A.} We start from the first two layers. Given that $U_1 U_2$ generates an anomaly free symmetry, there should exist a finite-depth local unitary that maps $U_1 U_2$ to an onsite operator \footnote{It was realized very recently that there exist lattice symmetries which are free of field-theoretic 't Hooft anomaly but are still not onsiteable \cite{Levin2025LatticeAnomaly,Else2025LatticeAnomaly,Thorngren2025LatticeAnomaly}. We do not have such problem here. }. The following unitary does the job: 
\begin{align}
    M_{1,2}:=\prod_{\Delta}^\text{(layers 1 \& 2)}{\rm CCZ}_\Delta=M_{1,2}^\dagger. 
\end{align}
Here, the product runs over all triangular plaquettes illustrated in Fig.\,\ref{fig:TriangularPlaquettes} -- both lower-left and upper-right triangles. ${\rm CCZ}_\Delta$ acts on the three vertices of the triangle $\Delta$ and is defined by ${\rm CCZ}_{i,j,k}=\Lambda_i(\CZ_{j,k})$. Note that ${\rm CCZ}_{i,j,k}$ is invariant under the permutation of site indices and hence there is no need to specify a site ordering. One can verify 
\begin{align}
    M_{1,2}(U_1 U_2)M_{1,2}=\prod_{n=1,2}\prod_jX_{n,j}. 
\end{align}
The $ZZ$ terms in the Hamiltonian are invariant under this unitary transformation. 

\textbf{Step 1B.} We can now gauge the onsite symmetry generated by $M_{1,2}(U_1 U_2)M_{1,2}$, treating it as a 1$d$ symmetry. We introduce gauge field sites at $n=3/2$ and $j=1/2,3/2,\cdots,N-1/2$, as in Fig.\,\ref{fig:2dTCLatticeGrowth}b. 
Eigenvalues of the following Gauss's law stabilizers will be enforced to be $+1$.  
\begin{align}
    \includegraphics[valign=c,scale=\ScalingFactor]{GaussLaw12.pdf}
    \label{eq:GaussLaw12}
\end{align}
$Z_{n,j}Z_{n,j+1}$ operators on the first two layers (which enter $H_0$ with a minus sign) are modified to the following form:  
\begin{align}
    \includegraphics[valign=c,scale=\ScalingFactor]{GaugedIsingTerms.pdf}. 
    \label{eq:GaugedIsingTerms}
\end{align}
An important difference between the present example and the previous toric code case is that the boundary symmetry operator $U_2$ also gets modified by the gauging. To see how this happens, we first compute
\begin{align}
    &M_{1,2}U_2M_{1,2}=\nonumber\\
    &(-1)
    \includegraphics[valign=c,scale=\ScalingFactor]{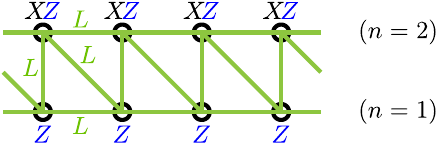}.
\end{align}
Here, each green bond represents an $L$ operator as defined in \eqref{eq:LjkDef} acting on the two sites it connects. As a remark about operator ordering, here and throughout, when we draw an operator that is a product of Pauli-$X$ and $Z$-basis diagonal operators, the Pauli-$X$ operators are always on the left while the diagonal operators are always on the right. The operator in the equation above does not commute with the Gauss's law stabilizers from gauging $U_1 U_2$, and hence it is no longer a valid symmetry after the gauging. A natural solution to this issue is to also couple $M_{1,2}U_2M_{1,2}$ to the gauge field. This is not hard: All we need is to apply the gauging map to the $L$ operators. Each $L$ operator contains a $ZZ$ operator exponentiated. Hence, the gauging transformations of $L$ operators follow from those of the $ZZ$ operators, such as $Z_{2,j}Z_{1,j+1}\mapsto Z_{2,j}Z_{3/2,j+1/2}Z_{1,j+1}$ and those shown in \eqref{eq:GaugedIsingTerms}. $Z_{2,j}Z_{1,j}$ need not be modified as it is already gauge invariant. We will denote by $V_2$ the result of the above procedure applied to $M_{1,2}U_2M_{1,2}$; let us not explicitly write it down. $V_2$ is a valid $\mathbb{Z}_2$ symmetry operator of our theory and our goal next is to gauge the symmetry generated by $V_2 U_3$. 

\Edited{As a side remark, the $\mathbb{Z}_2$ symmetry generated by $V_2 U_3$ is indeed anomaly free and can be gauged. We will find an explicit finite-depth local unitary transformation that maps it to an onsite symmetry. One might also try to directly compute the Else-Nayak anomaly index \cite{ElseNayak} of the symmetry generated by $V_2$ and the result will be the same as that of $U_2$. A brief argument is as follows: $M_{1,2}U_2M_{1,2}$ is a finite-depth unitary circuit consisting of Pauli's and $L$ operators. $V_2$ has an inherited circuit structure which may be used to compute the anomaly index. If we restrict to the subspace of trivial gauge field configuration ($Z=1$), the circuit of $V_2$ reduces to that of $M_{1,2}U_2M_{1,2}$. Consequently, $V_2$ has the same anomaly index as $M_{1,2}U_2M_{1,2}$ and thus also as $U_2$ because the anomaly index is invariant under a finite-depth local unitary transformation.} 

\textbf{Step 1C.} It will be helpful to simplify $V_2$ a bit. A natural attempt is to apply the unitary $M_{1,2}$ again in order to cancel some of its effect on $U_2$ earlier. The result is 
\begin{align}
    &U_2':=M_{1,2}V_2M_{1,2}\nonumber\\
    &=(-1)\includegraphics[valign=c,scale=\ScalingFactor]{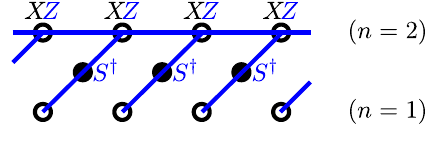}. 
    \label{eq:U2primed}
\end{align}
Here and throughout, a blue bond without any labeling represents a $\CZ$ operator acting on the two sites it connects. The following relation was used in the derivation of the above equation: 
\begin{align}
    i^{(1-Z_1Z_2Z_3)/2}=S_1S_2S_3\CZ_{1,2}\CZ_{2,3}\CZ_{3,1}. 
    \label{eq:tildeLrelation}
\end{align}
Hamiltonian terms such as those in \eqref{eq:GaugedIsingTerms} are unaffected by this unitary transformation. However, the Gauss's law stabilizers are modified to the following form.  
\begin{align}
    \includegraphics[valign=c,scale=\ScalingFactor]{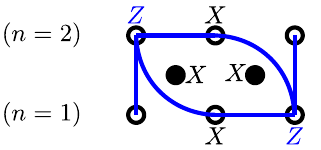}
    \label{eq:GaussLaw12Transformed}
\end{align}
Again, blue bonds represent $\CZ$ operators. 

\textbf{Step 2A.} Let us now move upward to layers $2$ and $3$. We would like to gauge $U_2' U_3$ which looks like 
\begin{align}
    \includegraphics[valign=c,scale=\ScalingFactor]{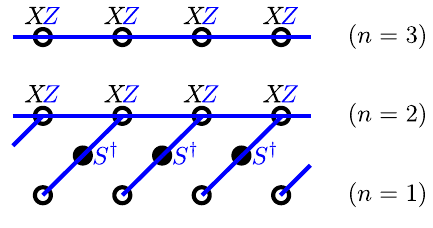}. 
\end{align}
Compared to $U_2U_3$, there are some extra ``legs'' in between layers $1$ and $2$, and hence some extra efforts are needed to transform this symmetry generator to an onsite form. We define a unitary operator $\Xi_{2,3}$ which is the product of the ingredients listed below.  
\begin{itemize}
    \item $M_{2,3}$, defined analogously to $M_{1,2}$ but shifted upward by one layer. 
    \item ${\rm CCZ}$ gates acting on the following triples of sites: $(n=1,j)$, $(n=3/2,j+1/2)$, and $(n=2,j+1)$, for all $j$. 
    \item Gates ${\rm CS}_{ab}:=\Lambda_a(S_b)$ on the following pairs of sites: $[a,b]=[(1,j),(3/2,j+1/2)]$ and $[(2,j+1),(3/2,j+1/2)]$ for all $j$. 
\end{itemize}
With this definition, one can verify that 
\begin{align}
    \Xi_{2,3}(U_2' U_3)\Xi_{2,3}^\dagger=\prod_{n=2,3}\prod_j X_{n,j}. 
\end{align}
No Hamiltonian term is affected by the unitary transformation $\Xi_{2,3}$. In contrast, the previous Gauss's law stabilizers are modified to the following form. 
\begin{align}
    (-1)\includegraphics[valign=c,scale=\ScalingFactor]{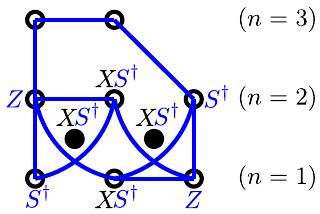}
    \label{eq:GaussLaw12_2ndTransf}
\end{align}

\textbf{Step 2B.} We can now gauge the onsite symmetry generated by $\Xi_{2,3}(U_2' U_3)\Xi_{2,3}^\dagger$ as in Step 2A. We get new Gauss's law stabilizers similar to those in \eqref{eq:GaussLaw12} but shifted upward by one layer. Hamiltonian terms with support on the 1st layer are unaffected, but those with support on layers $2$ and $3$ are mapped to
\begin{align}
    \includegraphics[valign=c,scale=\ScalingFactor]{GaugedIsingTerms23.pdf}. 
\end{align}
The hardest calculation of this step is to apply the gauging map to the old Gauss's law stabilizers in \eqref{eq:GaussLaw12_2ndTransf}. In general, given a symmetric local operator $O$, the gauged operator $O'$ is the unique local operator satisfying two properties: 
\begin{itemize}
    \item $O'$ commutes with new Gauss's law operators. 
    \item $O'$ reduces to $O$ if the gauge potentials on the gauge field sites are trivial, namely $Z_{5/2,j}=1$ for all $j$ in our case. 
\end{itemize}
For a proof, see the proposition at the end of Appendix \ref{app:GaugingReview}. 
Hence, the strategy is as follows: We add some controlled gates to the operator $O$, with the gauge field qubits being the controlling qubits, so that the modified operator $O'$ commutes with all new Gauss's law stabilizers. This approach is much simpler than expanding $O$ in the Pauli basis. As a result, we find that the old Gauss's law stabilizers are transformed via the gauging map to 
\begin{align}
    (-1)\includegraphics[valign=c,scale=\ScalingFactor]{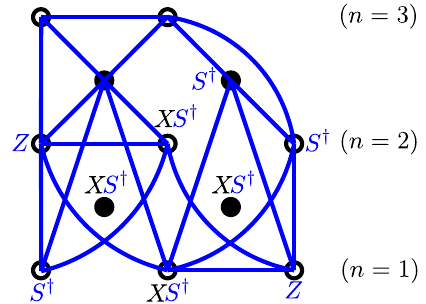}. 
    \label{eq:GaussLaw12Gauged}
\end{align}
The boundary symmetry  $\Xi_{2,3}U_3\Xi_{2,3}^\dagger=M_{2,3}U_3M_{2,3}$ just looks the same as $M_{1,2}U_2M_{1,2}$ up to a shift in layer indices. Hence, after gauging, we get $V_3$ that is just the shifted version of $V_2$. 

\textbf{Step 2C.} Similar to Step 1C, we apply the unitary $M_{2,3}$ in order to simplify the boundary symmetry $V_3$. As a result, we get $U_3'$ that is similar to $U_2'$ in \eqref{eq:U2primed}. Many other things also get transformed: The new Gauss's law stabilizers, originating from the previous step, now looks similar to those in \eqref{eq:GaussLaw12Transformed} modulo a shift in layer indices. The older Gauss's law stabilizers shown in \eqref{eq:GaussLaw12Gauged} are modified as well but since this is not a hard calculation, let us not write down the result explicitly. Hamiltonian terms are all diagonal operators and are hence unaffected. 

\textbf{Step 3.} We shall now repeat Steps 2A-2C for layers $3$ and $4$: Apply $\Xi_{3,4}$, gauge, and then apply $M_{3,4}$. The oldest Gauss's law stabilizers originating from Step 1 reaches their final form: See Fig.\,\ref{fig:DSBulkStabs}a and regard the bottom layer of the plotted operator as $n=1$. 
These will not be further modified in subsequent steps. 
All other changes, including the introduction of new Gauss's law terms, modifications to Hamiltonian terms, and transformations of old Gauss's law stabilizers from Step 2, are the same as what we have seen in Step 2 except for a shift in layer indices. 

Subsequent steps for higher layers will all be similar and we will not state them explicitly. We will not bother to consider what happens near the top layer, or we may imagine an infinite system with no upper bound on $n$. 

\begin{figure*}
    \centering
    \includegraphics[scale=\ScalingFactor]{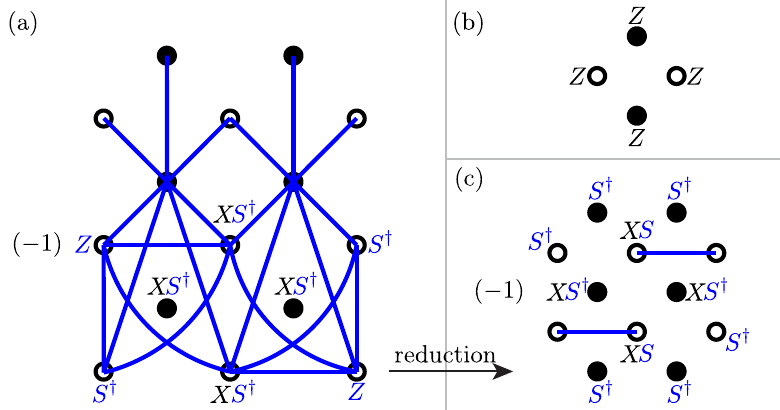}
    \caption{$2d$ bulk stabilizers constructed from the $1d$ anomalous $\mathbb{Z}_2$ symmetry. Each blue bond represents a controlled-$Z$ operator. 
    Stabilizers in (a) are denoted by $S^E_{n+1/2,j}$ where $(n+1/2,j)$ locates the center of the four Pauli-$X$. Stabilizers in (b) are denoted by $S^M_{n,j+1/2}$ where $(n,j+1/2)$ locates the center of the four Pauli-$Z$. When stabilizers in (b) are satisfied, (a) can be reduced to the form of (c) which we denote by $S^{ER}_{n+1/2,j}$. 
    Either (a) and (b), or (c) and (b), lead to the double semion topological order. }
    \label{fig:DSBulkStabs}
\end{figure*}

With all those efforts, we have obtained a new $2d$ theory. The Hilbert space of the theory is constrained by the Gauss's law: Stabilizers as shown in Fig.\,\ref{fig:DSBulkStabs}a are required to have $+1$ eigenvalues. The Hamiltonian stabilizers, which enter the Hamiltonian with a minus sign, are all products of Pauli-$Z$ and look like Fig.\,\ref{fig:DSBulkStabs}b away from the boundary. For later convenience, let us assign some names to those stabilizers: We denote by $S^E_{n+1/2,j}$ the stabilizers in Fig.\,\ref{fig:DSBulkStabs}a where the subscripts indicate the center of the four Pauli-$X$ operators therein. We denote by $S^M_{n,j+1/2}$ the stabilizers in Fig.\,\ref{fig:DSBulkStabs}b where the subscripts indicate the center of the four Pauli-$Z$ operators. Although not obvious, one can check that \footnote{This is actually a consequence of our gauging procedure. See the proposition at the end of Appendix \ref{app:GaugingReview}. } $(S^E_{n+1/2,j})^2=1$, and thus the eigenvalues of each $S^E$ stabilizer can only be $\pm 1$. $S^E$ is unitary since it is the product of multiple unitary operators. It follows that $S^E$ is also hermitian. The Gauss's law can be expressed as $S^E_{n+1/2,j}=1$ for all $n\geq 1$ and $j$. 
The Hamiltonian, denoted by $\tilde H$, reads
\begin{align}
    \tilde H= -\sum_{n\geq 1,j} S^M_{n,j+1/2}. 
\end{align}
Here, $S^M_{n=1,j+1/2}$ for each $j$ acts on one site outside the system (it contains $Z_{-1/2,j+1/2}$) and our convention is to just remove its action on this non-existing site; see the second figure in Eq.\,\ref{eq:GaugedIsingTerms}. Equivalently, we may imagine adding a row of auxiliary sites at $n=-1/2$ and imposing $Z_{-1/2,j+1/2}=1$ for all $j$. Then there is no need to truncate $S^M_{n=1,j+1/2}$. 

\subsection{Ground States and Simplified Stabilizers}

We now move on to determine the bulk topological order of the new theory $\tilde H$. To this end, it is convenient to consider periodic boundary conditions in both spatial directions. Namely, we will make the identifications $j\sim j+N_1$ and $n\sim n+N_2$ for some positive integers $N_1,N_2$. The total number of qubits in the system is then $N=2N_1N_2$. In this subsection, we will show that on such torus geometries, our model is gapped and has four-fold degenerate ground states. Moreover, both the energy gap and the ground state degeneracy are stable against perturbations, implying a nontrivial topological order. We will also show that the rather complicated stabilizers in Fig.\,\ref{fig:DSBulkStabs}a may be replaced by those in Fig.\,\ref{fig:DSBulkStabs}c without altering the topological order. 

To proceed, it is necessary to do some calculations with the stabilizers in Fig.\,\ref{fig:DSBulkStabs}, which could be annoying. There is actually a nice method to greatly simplify the work. The key observation is that all stabilizers in Fig.\,\ref{fig:DSBulkStabs} belong to the following operator group:
\begin{align}
    \left\{ \lambda\prod_{a\in\mc I}X_a\prod_{\{b,c\}\in\mc J}\CZ_{b,c}\prod_{(k,d)\in \mc K}S^k_d  \right\}, 
\end{align}
where $\lambda$ is a nonzero complex coefficient, $a,b,c,d$ are all site indices, and $k\in\mathbb{Z}_4$. Saving an operator from this group in a computer is memory efficient -- one just need to record $\lambda$ and the index sets $\mc I,\mc J,\mc K$. Computing operator multiplication and inverse are also numerically efficient. Many of the results presented below are actually obtained numerically using this approach. 

\textbf{No Frustration. }All the $S^E$ and $S^M$ stabilizers commute with each other, which is a consequence of our gauging procedures \footnote{See the proposition at the end of Appendix \ref{app:GaugingReview}. } and can also be directly checked. As a result,  energy eigenstates of the theory can be taken to be simultaneous eigenstates of all stabilizers. We shall first show that the theory is frustration free: The ground states satisfy all stabilizers. 

Let $\ket{0}$ and $\ket{1}$ be the two eigenstates of Pauli-$Z$ with eigenvalues $\pm 1$, respectively. We consider the following unnormalized ground state candidate. 
\begin{align}
    \prod_{n,j}\left(\frac{1+S^E_{n+1/2,j}}{2}\right)\ket{0}^{\otimes N}
    \label{eq:DSGndState00}
\end{align}
As long as this state is non-vanishing, it must be a ground state because it by construction satisfies all the stabilizers. We can try to expand the product of projectors into a sum of monomials in the $S^E$ stabilizers and think about the action of each monomial on the all-up spin configuration $\ket{0}^{\otimes N}$. Recall that each $S^E$ stabilizer is the product of four Pauli-$X$ and certain $Z$-basis diagonal unitary operators. When acting on any $Z$-basis eigenstate, the diagonal unitaries will only leave a phase factor, while the four Pauli-$X$ operators will change the spin configuration on the corresponding sites. This suggests that the only spin all-up component of the above state is as follows: 
\begin{align}
    2^{-N_1N_2}(1+\prod_{n,j}S^E_{n+1/2,j})\ket{0}^{\otimes N}, 
\end{align}
which turns out to be non-vanishing because we have found  
\begin{align}
    \prod_{n,j}S^E_{n+1/2,j}=1. 
\end{align}
This implies that the state in \eqref{eq:DSGndState00} is indeed non-vanishing and is a valid ground state of the system. 

\textbf{Ground State Degeneracy. }Given any $Z$-basis eigenstate that satisfies all the $S^M$ stabilizers, applying the $S^E$ projectors as in \eqref{eq:DSGndState00} generates a ground state of the theory. Conversely, any ground state belongs to the space spanned by such projected states because we can expand any state as a superposition of $Z$-basis spin configurations. We claim that there are four independent ground states generated by the following states. 
\begin{align}
    \left(\prod_jX_{n_0,j}\right)^x\left(\prod_n X_{n+1/2,j_0+1/2}\right)^y\ket{0}^{\otimes N}, 
\end{align}
where $x,y\in\{0,1\}$ and $n_0,j_0$ are arbitrary fixed integers. To prove this statement, we observe the following fact: If two $Z$-basis eigenstates can be mapped to each other up to a phase by the action of several $S^E$ stabilizers, these two states will generate equivalent energy eigenstates after applying the $S^E$ projectors as in \eqref{eq:DSGndState00}. Then we apply an argument often used for the toric code case: $Z$-basis spin configurations satisfying all $S^M$'s can be regarded as loop configurations if we identify white (black) lattice sites as horizontal (vertical) links. The $S^E$ stabilizers can create or annihilate any contractible loops. Hence, the independent degrees of freedom labeling the different ground states are the presence/absence of non-contractible loops around the two periodic directions, as in the equation above. We have thus proved that the model has four-fold degenerate ground states. There is also obviously an energy gap above these states, given that all stabilizers can be simultaneously diagonalized. 

\textbf{Robustness against Perturbations. }
In order to confirm that our model indeed has a nontrivial topological order, we would like to show that the ground state degeneracy is not accidental or due to any conventional (0-form) symmetry breaking. To this end, we shall prove the following claim: Any local operator $O$ that commutes with all stabilizers in the theory can be generated by those stabilizers. This condition suggests that local perturbations are only able to lift the ground state degeneracy at $n$-th order in perturbation theory, where $n$ is about the linear size of the system. It is a commonly used condition for proving a nontrivial topological order, though there is a caveat to be mentioned below. 

A proof of this claim is sketched as follows. We can uniquely decompose the operator $O$ in the following form: 
\begin{align}
    O=\sum_A \lambda_A X_A D_A, 
\end{align}
where $A$ runs over all subsets of the lattice, $\lambda_A$ is a number coefficient, $X_A$ is the product of Pauli-$X$ on all qubits in $A$, and $D_A$ is a diagonal operator labeled by $A$ but not necessarily acting within $A$. Each term in the above decomposition must seperately commute with all stabilizers of the theory. Therefore, we may simply assume that $O=X_A D$ for some subset $A$ and diagonal operator $D$. In order to commute with the $S^M$ stabilizers, sites in $A$ must form loops when we regard white (black) sites as horizontal (vertical) links. Moreover, those loops are all contractible due to the locality of $O$. As a result, $X_A D$ can be expressed as a product of several $S^E$ stabilizers and some diagonal operator $D'$. We can uniquely decompose $D'$ as $D'=\sum_B\lambda'_B Z_B$ where similarly, $B$ runs over lattice subsets, $\lambda_B'$ is a number coefficient, and $Z_B$ is the product of Pauli-$Z$ on all qubits in $B$. Each term must seperately commute with all stabilizers. In order for each $Z_B$ with $\lambda_B'\neq 0$ to commute with the $S^E$ stabilizers, the sites in $B$ must form loops when we regard black (white) sites as horizonal (vertical) links. As a result, each such $Z_B$ is a product of several $S^M$ stabilizers and this completes the proof. 

Strictly speaking, the condition proved above is not sufficient to guarantee a nontrivial topological order. It is possible to construct an artificial model that satisfies the condition above but has an unstable energy gap \cite{Bravyi2010TQO}: The energy gap closes under a certain local perturbation whose local strength can be made arbitrarily small as the system size increases. It is possible to show that our model does not suffer from this problem by proving an additional condition given in Ref.\,\onlinecite{Bravyi2010TQO}. However, since such pathological examples are rare in practice, let us just stop here and not get into more technical details along this direction.


\textbf{Reduced Stabilizers. }Given the frustration free condition, we can simplify our model stabilizers if only the ground states or the topological phase of the system are concerned. In the subspace where all the stabilizers in Fig.\,\ref{fig:DSBulkStabs}b are satisfied, namely when $S^M_{n,j+1/2}=1$ for all $n$ and $j$, the stabilizers in Fig.\,\ref{fig:DSBulkStabs}a are equivalent to those in Fig.\,\ref{fig:DSBulkStabs}c. In obtaining this reduction, we have utilized the following relations under the condition $Z_1Z_2Z_3Z_4=1$:
\begin{itemize}
    \item $\CZ_{1,j}\CZ_{2,j}\CZ_{3,j}\CZ_{4,j}=1$ where $\CZ_{j,j}$ should be understood as $Z_j$ when $j\in\{1,2,3,4\}$. 
    \item $\CZ_{2,3}\CZ_{2,4}\CZ_{3,4}=S_1^\dagger S_2S_3S_4=S_1S_2^\dagger S_3^\dagger S_4^\dagger$, which may be derived using \eqref{eq:tildeLrelation}. 
\end{itemize}
Let us denote the stabilizers in Fig.\,\ref{fig:DSBulkStabs}c by $S^{ER}_{n+1/2,j}$ where again the subscripts indicate the center of the four Pauli-$X$ operators. We can thus replace the Gauss's law constraints $S^E=1$ by $S^{ER}=1$ without altering the ground states or the topological phase. 

We may soften the Gauss's law constraints by adding them to the Hamiltonian and enforcing them energetically, which will not alter the ground states but will allow for more general excitations. For example, using the original stabilizers $S^E$ and $S^M$, we may write
\begin{align}
    H=-\sum_{n,j}(S^{E}_{n+1/2,j}+S^M_{n,j+1/2}) 
\end{align}
and impose no hard constraint on the Hilbert space. 
It is a little trickier to add the $S^{ER}$ stabilizers to the Hamiltonian because they are not hermitian but instead satisfy $(S^{ER})^4=1$ or equivalently $(S^{ER})^\dagger=(S^{ER})^3$. One possibility is to write
\begin{align}
    H^R=-\sum_{n,j}\left[\frac{S^{ER}_{n+1/2,j}+(S^{ER}_{n+1/2,j})^\dagger}{2}
    +S^M_{n,j+1/2}\right]. 
    \label{eq:HDSReduced}
\end{align}
We emphasize that our previous proof of the robustness against perturbations apply to all those alternative forms of the theory. 

\subsection{Double Semion Topological Order}
\begin{figure}
    \centering
    \includegraphics[scale=\ScalingFactor]{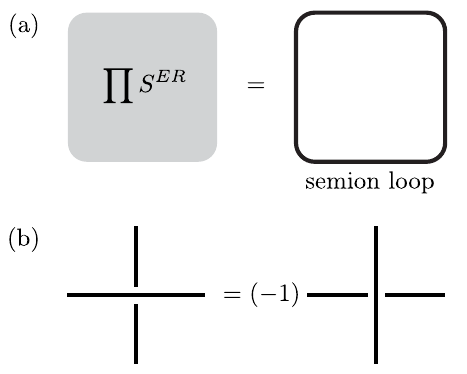}
    \caption{Schematic illustration of the semion line operators. (a) Taking the product of $S^{ER}$ stabilizers in a region generates a semion loop operator supported on the region boundary. (b) The semion statistics manifests in the braiding property of the line operators. }
    \label{fig:SemionLoopSchematic}
\end{figure}
We will now show that our model realizes the double semion topological order \cite{LevinWen2005}, which has the same torus ground state degeneracy as the toric code but has different anyon statistics. 

\textbf{Anyon Statistics. }In order to analyze anyon statistics, we need to construct line operators that can create, annihilate, and move anyons. Similar to the toric code case, such line operators can be found by multiplying stabilizers together. For concreteness, we will focus on the reduced stabilizers or the Hamiltonian in \eqref{eq:HDSReduced}. We have found that the product of all $S^{ER}$ stabilizers in a region, such as a rectangular illustrated in Fig.\,\ref{fig:SemionLoopSchematic}a, equals to a loop operator supported on the region boundary. The explicit form of the loop operator is given in Appendix \ref{app:DS} for readers interested in the details. An open segment of this loop operator will generate a dual pair of anyons at its open ends. It turns out that the upper and lower edges of the loop operator have the same form, and similarly, the left and right edges also look the same, hence this loop is unoriented and the corresponding anyon is its own antiparticle. We will denote by $m_a$ the type of anyons created by open segments of this loop operator. One can check that the $m_a$-type anyon lines satisfy an anticommutation relation as shown in Fig.\,\ref{fig:ERLoop}b, it follows that $m_a$ is a \emph{semion}: A full braiding of two $m_a$-type anyons (exchanging twice) gives a minus sign. 

One can similarly generate another type of loop operator by multiplying the $S^M$ stabilizers and the result is familiar from the toric code model. We denote by $e$ the corresponding anyon type which is also its own antiparticle. $e$ is a boson but a full braiding of $e$ and $m_a$ gives a minus sign due to a similar anticommutation relation. Let $m_b:=e\times m_a$ be the fusion of these two anyon types. It follows from our results above that $m_b$ is also a semion and a full braiding of $e$ and $m_b$ also leaves a minus sign. 

These results are sufficient to distinguish the topological order of our model to that of the toric code, and are consistent with the double semion topological order. 

\textbf{Anomalous Boundary Symmetry. }We can also verify that our model indeed enforces an anomalous boundary $\mathbb{Z}_2$ symmetry in the presence of a spatial boundary, which is our starting point. Let us still assume the horizontal site index $j$ to be periodic, but take an open boundary condition along the layering direction: restricting to $1\leq n\leq N_2$. Consider a Hamiltonian of the form $H_{\rm bulk}+H_{\rm bdry}$. We choose $H_{\rm bulk}$ to be 
\begin{align}
    H_{\rm bulk}=&-\sum_{n=1}^{N_2-1}\sum_j\left[\frac{S^{ER}_{n+1/2,j}+(S^{ER}_{n+1/2,j})^\dagger}{2}\right]\nonumber\\
    &-\sum_{n=2}^{N_2-1}\sum_j S^M_{n,j+1/2}  
\end{align}
where $S^{ER}_{n+1/2,j}$ for $n=1$ and $n=N_2-1$ should be truncated in order to be well defined. We let $H_{\rm bdry}$ be a sum of local operators that are supported near the boundary and commute with $H_{\rm bulk}$. One can check that the product of all $S^{ER}$ and all $S^M$ entering $H_{\rm bulk}$ equals to $U_1U_{N_2}$. This result suggests in the absence of bulk excitations, an anomalous $\mathbb{Z}_2$ symmetry is enforced on the effective $1d$ boundary theory. 

\textbf{Other Models in the Literature. }We shall mention that there exist other square lattice double semion models in the literature: A $\mathbb{Z}_4$ Pauli stabilizer model is given in Ref.\,\onlinecite{Ellison2021PauliTQD}. Another qubit lattice model is given in Ref.\,\onlinecite{Chen2021HigherCup}. The latter looks similar to our reduced model $H^R$, but it is unclear whether there is a simple way of relating them. 

\section{A Nonabelian Example}\label{sec:NonabelianExample}

In this section, we will consider the simplest nonabelian example. 
Our goal is to construct a $2d$ topological order that characterizes the physics of an onsite finite group symmetry $G$ in $1d$. We have outlined the generalized bulk construction procedure in Section \ref{sec:PrescriptionGeneralized} and we will now see how it works precisely. 

Consider a 1$d$ chain of sites labeled by $j=1,2,\cdots, N_1$. The local Hilbert space on each site will be spanned by the orthonormal group element basis states $\{ \ket{g}|g\in G \}$. We denote by $P_j(g):=\ket{g}_j\bra{g}_j$ the projector onto the group element state $\ket{g}_j$ on site $j$.  Let $L_j(h)$ and $R_j(h)$ be respectively the left and right regular representation operators acting on site $j$, defined as
\begin{align}
    L_j(h)\ket{g}_j=\ket{hg}_j,\quad R_j(h)\ket{g}_j=\ket{gh^{-1}}_j.
\end{align}
The $1d$ Hamiltonian that we will use is 
\begin{align}
    H_G=\sum_j\sum_{g\in G}\left[-JP_j(g)P_{j+1}(g) -\eta|G|^{-1}L_j(g)\right], 
\end{align}
with $J,\eta\geq 0$. We may also substitute $L_j(g)$ in the second term by $R_j(g)$ without changing the Hamiltonian, as $\sum_g L_j(g)=\sum_gR_j(g)$. 
This Hamiltonian is a generalization of the transverse-field Ising model $H_{{\rm Ising},1d}$ considered in Section \ref{sec:ElementaryExamples}. It has both a ``left symmetry'' $G_L$ generated by $\rho_L(g):=\prod_jL_j(g)$ and a ``right symmetry'' $G_R$ generated by $\rho_R(g):=\prod_jR_j(g)$. Importantly, $\rho_L(g)$ and $\rho_R(h)$ commute with each other for all $g,h\in G$. We will thus say that the system has a $G_L\times G_R$ symmetry, although its representation is not faithful when $G$ has a nontrivial center. 

For our purpose of constructing a bulk topological order, we focus on the case $J> 0$ and $\eta=0$. 
The vacuum states \footnote{We use the term ``vacuum states'' for ground states satisfying the cluster decomposition property, i.e. ``non-cat'' states. } of $H_G$ are then given by $\ket{\Psi_g}:=\bigotimes_j\ket{g}_j$ for all $g\in G$. Hence, either $G_L$ or $G_R$ has been spontaneously broken completely. However, a $G$ subgroup of $G_L\times G_R$ has been preserved: $\rho_L(ghg^{-1})\rho_R(h)$ for all $h\in G$ leaves the state $\ket{\Psi_g}$ invariant. Note that as a generic feature of nonabelian symmetry breaking, different vacuum states preserve distinct but isomorphic symmetry subgroups. 

Our generalized layered gauging procedure is to stack many copies of the above $1d$ model, labeled by $n$, and then sequentially gauge the bilayer $G$ symmetries generated by $\rho_{n,L}(g)\rho_{n+1,R}(g)$. To implement this procedure, let us first recall how to gauge a nonabelian symmetry $G$ in $1d$. In general, suppose we have some $1d$ onsite $G$ symmetry generated by $\rho(g)=\prod_jU_j(g)$ with $j=1,2,3,\cdots$ being the site index. The gauging procedure starts by introducing gauge field degrees of freedom at half-integer sites $j+1/2$. The Hilbert space of each gauge field site is spanned by the group element basis $\{ \ket{g}|g\in G \}$ and supports left and right actions as we have already seen. The Gauss's law operators, or gauge symmetry generators, are given by $S_j(g):=R_{j-1/2}(g)U_j(g)L_{j+1/2}(g)$ for all $j$ and $g\in G$. Similar to the abelian case, we shall then couple the Hamiltonian terms to the gauge field so as to respect the gauge symmetry, and finally restrict the Hilbert space to the gauge invariant subspace satisfying $S_j(g)=1$ for all $j$ and $g$. In our case, when considering the $G$ symmetry generated by $\rho_{n,L}(g)\rho_{n+1,R}(g)$, $U_j(g)$ will be $L_{n,j}(g)R_{n+1,j}(g)$. 

\begin{figure}
    \centering
    \includegraphics[scale=\ScalingFactor]{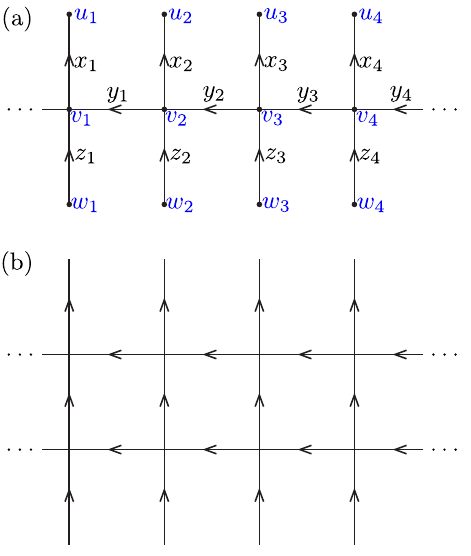}
    \caption{Layered gauging for the $1d$ onsite $G$-symmetry. (a) The first two layers and the corresponding gauge field degrees of freedom. (b) The first three layers. }
    \label{fig:GSymmTwoNThreeLayers}
\end{figure}

When dealing with nonabelian group element states, keeping track of left and right group actions can be quite cumbersome. There turns out to be a less painful and pictorial way of thinking about such states, which we shall introduce now. We associate each group spin, whose Hilbert space is spanned by the group element states $\{ \ket{g} \}$, to an \emph{oriented edge} (oriented link). We then define a rule: Reversing the orientation of a particular edge is equivalent to the basis change $\ket{g}\mapsto \ket{g^{-1}}$ on that edge. Let $e$ be an edge containing a vertex (site) $v$. With the above rule, we define a group multiplication operator $X_{e,v}(h)$ as follows: 
\begin{align}
    X_{e,v}(h)\bigg|\includegraphics[valign=c,scale=\ScalingFactor]{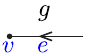}\bigg\rangle
    =\bigg|\includegraphics[valign=c,scale=\ScalingFactor]{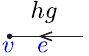}\bigg\rangle. 
\end{align}
In this equation, $e$ is an incoming edge with respect to $v$. The case of an outgoing edge follows directly from the orientation reversing rule. More explicitly, if we denote by $v_+(e)$ and $v_-(e)$ the starting and ending vertices of the edge $e$, then $X_{e,v_-(e)}(h)\ket{g}_e=\ket{hg}_e$ and $X_{e,v_+(e)}(h)\ket{g}_e=\ket{gh^{-1}}_e$. We see that left and right group multiplication operators have been unified to a single notation. 

Using this oriented edge representation, we can draw the lattice after gauging $\rho_{1,L}(g)\rho_{2,R}(g)$ as Fig.\,\ref{fig:GSymmTwoNThreeLayers}a. Here, $z_j$ ($x_j$) are the group elements on the first (second) layer and $y_j$ are the gauge field group elements. $u_j,v_j,w_j$ are labels for vertices which are purely geometric objects and contain no physical degrees of freedom. 
The Gauss's law can be written as 
\begin{align}
    A_{v_j}(h)=1\quad \text{for all $v_j$ and $h\in G$}, 
\end{align}
where the operator $A_v(h)$ for a vertex $v$ and a group element $h$ is defined as $A_v(h)=\prod_{e\supset v}X_{e,v}(h)$. Pictorially, 
\begin{align}
    A_v(h)\Bigg|\includegraphics[valign=c,scale=\ScalingFactor]{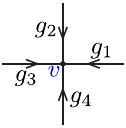}\Bigg\rangle=\Bigg|\includegraphics[valign=c,scale=\ScalingFactor]{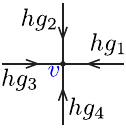}\Bigg\rangle. 
\end{align}
The original Hamiltonian terms are projectors onto the subspace where $x_jx_{j+1}^{-1}=1$ and $z_j^{-1}z_{j+1}=1$ (recall we set $\eta=0$). In the presence of gauge field, these need to be modified to
\begin{align}
    x_jy_jx_{j+1}^{-1}=1\quad\text{and}\quad z_j^{-1}y_jz_{j+1}=1
\end{align}
in order to preserve the gauge symmetry. In other words, after gauging, the Hamiltonian terms are updated to a new set of projectors which project onto the subspace with the above conditions. There is also a nice pictorial way of remembering these conditions. Given any oriented path on the lattice from one vertex to another, when the edges on the path all have consistent orientations with the path direction, we define the \emph{holonomy} of the path to the product of all group elements on the path in the reversed order. For example, in Fig.\,\ref{fig:GSymmTwoNThreeLayers}a, the holonomy of the path defined by the vertex sequence $(v_4,v_3,v_2,v_1)$ is $y_1y_2y_3$. 
With this definition, all the conditions above can be regarded as trivial holonomy conditions. For example, the holonomy of the path $(u_2,v_2,v_1,u_1)$ is $x_1y_1x_2^{-1}$. 

We can now proceed to gauge the symmetry generated by $\rho_{2,L}(g)\rho_{3,R}(g)$. This is possible thanks to the fact that $\rho_{2,L}(g)$ commute with the existing Gauss's law operators $A_{v_j}(h)$, which is also why we require the $1d$ model to have a $G_L\times G_R$ symmetry. After this gauging, the lattice becomes that in Fig.\,\ref{fig:GSymmTwoNThreeLayers}b. The Gauss's law constraints are $A_v(h)=1$ for all internal vertices $v$ and all $h\in G$. In the Hamiltonian, there are similar three-body terms at the topmost and bottommost layers, and also new four-body operators associated to the square plaquettes. More precisely, for each square plaquette, if we denote its boundary loop by $f$, there is a Hamiltonian term $-JB_f$ defined by the following equation: 
\begin{align}
    B_{f}\Bigg|\includegraphics[valign=c,scale=\ScalingFactor]{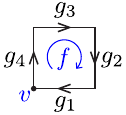}
    \Bigg\rangle
    =\delta_{g_1g_2g_3g_4,1}\Bigg|\includegraphics[valign=c,scale=\ScalingFactor]{LoopHolonomy.pdf}
    \Bigg\rangle.
    \label{eq:BfDef}
\end{align}
That is, $B_f$ projects onto the subspace of trivial holonomy around the loop $f$. Its definition does not depend on the base point $v$ of the loop or the loop orientation. 

We can continue this layered gauging procedure. Eventually, we obtain a bulk theory defined on a square lattice with group spins living on its links. Say we take the periodic boundary condition in both directions, the bulk Hamiltonian will take the following form: 
\begin{align}
    \tilde H= -J\sum_f B_f. 
\end{align}
Moreover, the Hilbert space of the theory is subject to the constraints $A_v(h)=1$ for all $v$ and $h\in G$. We may soften these constraints by adding the projectors $A_v:=|G|^{-1}\sum_{h\in G} A_v(h)$ to the Hamiltonian, as one may check that $A_v(g)A_v=A_v$. This gives the Hamiltonian
\begin{align}
    H=-\lambda\sum_v A_v-J\sum_f B_f
\end{align}
with $\lambda>0$. 
This is the famous quantum double model \cite{KitaevTC}, which generalizes the toric code and realizes a nonabelian topological order when $G$ is a nonabelian group. The model is frustration free and hence its ground states and topological properties do not depend on the precise magnitudes of $\lambda$ and $J$. 

\section{A Noninvertible Example}\label{sec:NoninvertibleExample}
Finally in this section, we will consider the simplest example with a noninvertible (fusion category) symmetry. This example is closely related to the one in Section \ref{sec:NonabelianExample} and we will also use some of the notations introduced there. 

Our starting point is again a 1$d$ periodic chain of group spins labeled by $j=1,2,\cdots,N_1$. The local Hilbert space of each group spin is spanned by the group element states $\{\ket{g}|g\in G\}$ for some finite group $G$. The Hamiltonian we shall use is 
\begin{align}
    H_{{\rm Rep}(G)}=\sum_j\left[ -\frac{\eta}{|G|}\sum_{g\in G}R_j(g)L_{j+1}(g)-J\ket{1}_j\bra{1}_j\right] 
\end{align}
with $\eta,J\geq 0$. Here, $L_j(g)$ and $R_j(g)$ as already defined in Section \ref{sec:NonabelianExample} generate left and right regular representations on site $j$, respectively. That is, $L_j(g)$ (respectively $R_j(g)$) acts as left (right) multiplication by $g$ ($g^{-1}$). $1\in G$ is the identity group element. Using the oriented link representation of group spins as introduced in Section \ref{sec:NonabelianExample}, we draw our 1$d$ lattice in the form of Fig.\,\ref{fig:RepGSymmExample}a. Here, $g_i$ label the group element basis states living on the oriented links. $v_i$ label the vertices which are purely geometric objects containing no physical degrees of freedom. With this picture, the operator $R_j(g)L_{j+1}(g)$ in the Hamiltonian can also be written as $A_{v_j}(g)$ which is defined in Section \ref{sec:NonabelianExample}. 

\begin{figure}
    \centering
    \includegraphics[scale=\ScalingFactor]{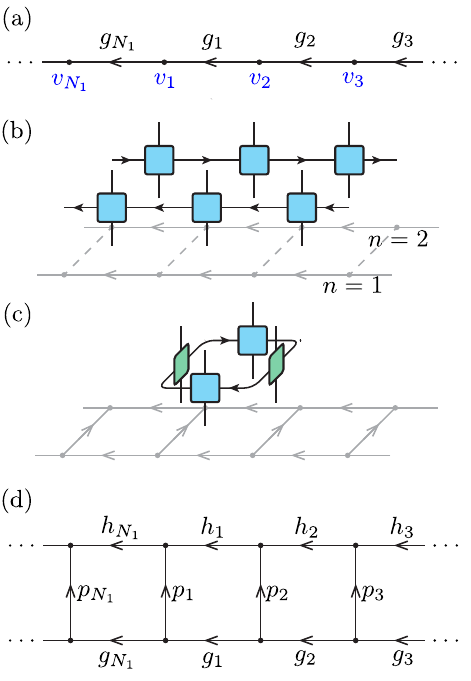}
    \caption{Layered gauging for the $1d$ ${\rm Rep}(G)$ symmetry. (a) Lattice geometry of a single layer. (b) Tensor network picture of the bilayer symmetry operators. Dashed lines on the lattice are to guide the eyes and have no physical meaning at the moment. (c) Local gauge symmetry generators from the gauging-like procedure. (d) Lattice geometry after the``gauging'' for the first two layers. }
    \label{fig:RepGSymmExample}
\end{figure}

We observe that (1) the first term of $H_{{\rm Rep}(G)}$ preserves the conjugacy class of $\phi:=g_1g_2\cdots g_{N_1}$ and (2) the second term of $H_{{\rm Rep}(G)}$ is diagonal in the group element basis. Therefore, any diagonal operator that is a function of the conjugacy class of $\phi$ commutes with the Hamiltonian and serves as a generalized symmetry operator. A simple basis for such operators are the conjugacy class projectors $\Pi_C$ defined as follows: $\Pi_C$ evaluates to $+1$ if the conjugacy class of $\phi$ is $C$ and evaluates to zero otherwise. There is also an alternative basis that is more convenient for our purpose: 
\begin{align}
    &\mc N_\mu:=\sum_{\{g_j\}}\nonumber\\
    &\chi_\mu(g_1g_2\cdots g_{N_1})\ket{g_1,g_2,\cdots ,g_{N_1}}\bra{g_1,g_2,\cdots ,g_{N_1}},  
    \label{eq:NmuDef}
\end{align}
where $\mu$ labels irreducible representations (irreps) of the group $G$ and $\chi_\mu(\cdot)$ is the character of the irrep $\mu$. Owing to the fact that irrep characters form a basis for complex class functions, $\{\mc N_\mu\}$ and $\{ \Pi_C \}$ span the same operator space. A nice property of the $\mc N_\mu$ operators is that they are \emph{matrix product operators}, which are generalizations of tensor product operators. More explicitly, if we denote by $D^{\mu}(g)$ the representation matrix for $g$ in the irrep $\mu$, we can write
\begin{align}
    \mc N_\mu=\sum_{\{g_j\}}\Tr[D^\mu(g_1)D^\mu(g_2)\cdots D^\mu(g_{N_1})]\bigotimes_j(\ket{g_j}\bra{g_j}). 
\end{align}
Pictorially, let 
\begin{align}
    \includegraphics[valign=c,scale=\ScalingFactor]{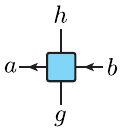}=\delta_{g,h}D^\mu_{ab}(g), 
\end{align}
where the irrep index $\mu$ is not explicitly shown in the picture. The operator ${\mc N}_\mu$ can be represented by 
\begin{align}
    &\bra{h_1,h_2,h_3,\cdots}\mc N_\mu\ket{g_1,g_2,g_3,\cdots}=\nonumber\\
    &\quad\includegraphics[valign=c,scale=\ScalingFactor]{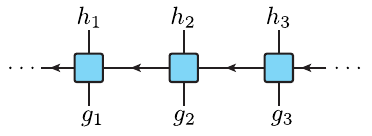}, 
\end{align}
where the horizontal legs of the tensors, also called virtual legs, have been contracted. It is not hard to check from \eqref{eq:NmuDef} that the $\mc N_\mu$ operators satisfy the same fusion rule as the irreps. That is, if 
\begin{align}
    [\mu]\otimes[\nu]=\bigoplus_\lambda C_{\mu\nu}^{\lambda}[\lambda], 
\end{align}
we also have 
\begin{align}
    \mc N_\mu\mc N_\nu=\sum_\lambda C_{\mu\nu}^\lambda \mc N_\lambda. 
    \label{eq:RepGSymmetryFusionRule}
\end{align}
For this reason, it is customary to say that $\{ \mc N_\mu \}$ generates a ${\rm Rep}(G)$ symmetry, which is an example of a noninvertible symmetry or fusion category symmetry. For later convenience, we also define the orientation reversed counterpart of $\mc N_\mu$ (the orientation reversing rule in Section \ref{sec:NonabelianExample} is used): 
\begin{align}
    \bar{\mc N}_\mu&:=\sum_{\{g_j\}}\chi_\mu(g_{N_1}^{-1}\cdots g_2^{-1}g_1^{-1})\ket{\{g_j\}}\bra{\{g_j\}}\\
    &=\sum_{\{g_j\}}\bar\chi_\mu(g_1g_2\cdots g_{N_1})\ket{\{g_j\}}\bra{\{g_j\}}=\mc N_{\bar \mu}, 
\end{align}
where $\bar\mu$ denotes the conjugate representation of $\mu$. $\bar{\mc N}_\mu$ is analogous to the inverse operator in invertible symmetries because $[\mu]\otimes [\bar\mu]$ contains the trivial representation. 

As a side note, if we gauge the $G_L$ symmetry of the model $H_G$ introduced in Section \ref{sec:NonabelianExample}, the resulting gauge theory is equivalent to $H_{{\rm Rep}(G)}$. To show this explicitly, one needs to apply a bunch of ${\rm CNOT}$-like unitary gates which localize the Gauss's law operators on single sites. Those redundant sites can then be removed and the obtained unconstrained theory is exactly $H_{{\rm Rep}(G)}$. We will not use this result but we hope it provides a natural way of writing down the Hamiltonian $H_{{\rm Rep}(G)}$. 

For our bulk construction purpose, we focus on $\eta>0$ and $J=0$. The vacuum states of $H_{{\rm Rep}(G)}$ with those parameters are as follows \footnote{One can verify the cluster decomposition property of these states using a matrix product state representation and the orthogonality relation \eqref{eq:OrthRelation}. }:
\begin{align}
    \ket{\Psi_\mu}=\sum_{\{g_j\}}\chi_\mu(g_1g_2\cdots g_{N_1})\ket{g_1,g_2,\cdots,g_{N_1}}. 
\end{align}
These states spontaneously break the ${\rm Rep}(G)$ symmetry since 
\begin{align}
    \mc N_\mu\ket{\Psi_\nu}=C_{\mu\nu}^\lambda\ket{\Psi_\lambda}. 
\end{align}

Mimicking our layered gauging construction for abelian invertible symmetries, we would like to stack many copies of the $1d$ model $H_{{\rm Rep}(G)}$, labeled by $n=1,2,3,\cdots$, and then sequentially gauge the bilayer symmetries generated by $\mc N_{n,\mu}\bar{\mc N}_{n+1,\mu}$. 
See Fig.\,\ref{fig:RepGSymmExample}b for an illustration of $\mc N_{1,\mu}\bar{\mc N}_{2,\mu}$ acting on our 2$d$ layered lattice. An immediate issue we encounter is that the operators $\{ \mc N_{n,\mu}\bar{\mc N}_{n+1,\mu} \}$ do not generate a ${\rm Rep}(G)$ symmetry, as the fusion rule \eqref{eq:RepGSymmetryFusionRule} is not group like. Therefore, it may sound weird to gauge the whatever symmetry generated by $\{ \mc N_{n,\mu}\bar{\mc N}_{n+1,\mu} \}$. We note that this issue may be regarded as a problem with the boundary condition: If we cut open the $\mc N_\mu$ and $\bar{\mc N}_\mu$ loops and join them together to form a single new loop, then these loop operators will satisfy the ${\rm Rep}(G)$ fusion rule. Despite this apparent issue, we found that there is a natural gauging-like procedure to promote the global symmetry operators $\mc N_{n,\mu}\bar{\mc N}_{n+1,\mu}$ to local ``gauge'' symmetry generators. This is illustrated in Fig.\,\ref{fig:RepGSymmExample}c. That is, we introduce new ``gauge field'' degrees of freedom as ladder rungs of the lattice and then make up small closed loops using the local tensors in the symmetry operators $\mc N_{n,\mu}$, $\bar{\mc N}_{n+1,\mu}$. What local Hilbert space should we choose for the gauge field? From the figure, we see that each gauge field link should support local tensors contractible with those in $\mc N_{n,\mu}$ and $\bar{\mc N}_{n+1,\mu}$. Hence, the most natural choice is to also put group spins on those links and use the same set of tensors to form local loop operators. Applying this procedure to the first two layers, the lattice geometry becomes that in Fig.\,\ref{fig:RepGSymmExample}d. 
The gauge symmetry generators, or Gauss's law operators, are given by: 
\begin{align}
    &S^{(1,j)}_{\mu}:=\sum_{g_j,h_j,p_j,p_{j+1}}\nonumber\\
    &\chi_\mu(g_jp_{j+1}^{-1}h_j^{-1}p_j)\ket{g_j,h_j,p_j,p_{j+1}}\bra{g_j,h_j,p_j,p_{j+1}},  
\end{align}
where the group element labels are shown in Fig.\,\ref{fig:RepGSymmExample}d. Note that these operators \emph{do} satisfy the ${\rm Rep}(G)$ fusion rule. 
Using Schur's orthogonality relation
\begin{align}
    \sum_{g\in G}\bar D^\mu_{ab}(g)D^\nu_{cd}(g)=\frac{|G|}{d_\mu}\delta_{\mu\nu}\delta_{ac}\delta_{bd}, 
    \label{eq:OrthRelation}
\end{align}
one can also check that when all the gauge field links are in the state $\ket{+}:=|G|^{-1}\sum_{g\in G}\ket{g}$, $\prod_{j}S^{(1,j)}_\mu$ is proportional to $\mc N_{n,\mu}\bar{\mc N}_{n+1,\mu}$. 
Shall we then demand $S^{(1,j)}_\mu=1$ as the Gauss's law? This is actually not possible because it contradicts the operator fusion rule. In other words, $[\mu]\mapsto 1$ is not a valid representation of the ${\rm Rep}(G)$ category. We should instead regard $[\mu]\mapsto d_\mu$ as the analog of trivial representation in the group case, where $d_\mu$ is the Hilbert space dimension of the irrep $\mu$. We therefore declare the following Gauss's law: 
\begin{align}
    S^{(1,j)}_\mu=d_\mu
\end{align}
for all $j$ and $\mu$. This is actually equivalent to the trivial flux condition: 
\begin{align}
    g_jp_{j+1}^{-1}h_j^{-1}p_j=1
\end{align}
for all $j$. 

Next, we shall couple the Hamiltonian terms to the gauge field so as to respect the gauge symmetry. This is not hard in our case. Recall that the original Hamiltonian (after setting $J=0$) consists of two-body operators $|G|^{-1}\sum_g A_v(g)$. After gauging, we can still write the modified operators as $|G|^{-1}\sum_g A_v(g)$, but these are now three-body operators acting on the new lattice in Fig.\,\ref{fig:RepGSymmExample}d. One can check that when the gauge field links are in the $\ket{+}$ states, those modified operators reduce to their original form. 

Similar gauging-like procedures can be implemented for other pairs of layers. We will skip the details and directly present the final result. The final bulk theory is defined on a square lattice with group spins living on the links. If we take the periodic boundary condition in the layering direction as well, the Hamiltonian is 
\begin{align}
    \tilde H=-\eta \sum_v A_v
\end{align}
where the summation is over all vertices and $A_v:=|G|^{-1}\sum_{g\in G}A_v(g)$ are now four-body operators according to the square lattice geometry. In addition, there is a Gauss's law constraint $B_f=1$ for each plaquette $f$, where $B_f$ as defined in \eqref{eq:BfDef} projects onto the trivial flux subspace. If we soften those Hilbert space constraints, we will again obtain the quantum double model: 
\begin{align}
    H= -\eta\sum_v A_v-\lambda \sum_f B_f. 
\end{align}
We have recovered the well-known result that both $G$ and ${\rm Rep}(G)$ symmetries in $1d$ correspond to the quantum double topological order in $2d$. 

\Edited{In the context of topological holography, it is believed that \emph{dual pairs of generalized symmetries} related by gauging correspond to the same bulk topological order; see e.g. Refs.\,\onlinecite{Chen2025Essay,Ji201912NoninvAnomaly}. The $1d$ $G$ and ${\rm Rep}(G)$ symmetries is one example of a dual pair of symmetries, meaning that upon gauging a $G$ symmetry, we will necessarily end up getting a ${\rm Rep}(G)$ symmetric theory and vice versa. The $2d$ conventional $\mathbb{Z}_2$ and one-form $\mathbb{Z}_2$ symmetries considered respectively in Sections \ref{sec:3dTCExample} and \ref{sec:HigherFormExample} is another example, and indeed they both lead to the $3d$ toric code bulk topological order. }

\section{Discussion}\label{sec:Discussion}
To summarize, we propose layered gauging as a new method for constructing $(k+1)$-dimensional topological orders from $k$-dimensional generalized symmetries. The generality of the method is demonstrated through a number of examples in different spatial dimensions and with various symmetry types.  

The examples that we have considered so far are still limited, especially in the cases of anomalous, nonabelian, and noninvertible symmetries. \Edited{The effects of different initial Hamiltonians also need a better understanding.} Therefore, more examples are needed in future works to further test our method. We hope that during this process, our method can give rise to more interesting lattice models with topological orders. 

\Edited{Throughout this work, we focus on the generation of $(k+1)$-dimensional topological orders from $k$-dimensional symmetries. An important next question is whether this layered gauging framework can also be used to understand the $k$-dimensional boundary theories of the $(k+1)$-dimensional topological orders. For example, whether it can be used to construct or characterize gapped boundaries of the resulting bulk topological orders. }

Under certain circumstances, the process of gauging can be efficiently realized in quantum experiment platforms by combining unitary gates, measurements, and feedforward \cite{Nat2021KWfromMeasurement,Ashkenazi2022Duality,Bravyi2022PrepareSolvableQD}. Hence, another future direction is to explore whether our layered gauging construction can lead to new state preparation protocols. 

From the perspective of QEC codes, our bulk construction method may be regarded as a generalization of the hypergraph product \cite{Zemor2009HGP}. More precisely, the prescription of gauging bilayer symmetries acting on nearest-neighbor pairs of layers match with the checks of the repetition code (equivalently, spin interaction terms of the $1d$ classical Ising model). Therefore, each topological order constructed in this paper may be regarded as a product between the repetition code ($1d$ classical Ising model) and another quantum lattice model with spontaneous symmetry breaking. In fact, this relation between the hypergraph product and layered gauging is already discussed in Ref.\,\onlinecite{Tantivasadakarn2026HGP} in the case of CSS bulk topological orders. 
It will be interesting to see whether this relation can be made precise in more general situations and whether new QEC codes beyond the CSS type can be produced. For example, replacing the repetition code by some other classical error correcting codes, we may define a more general layered gauging prescription. The Hamiltonian on each layer may also become nonlocal if we aim at QEC codes with better performance.

\section*{Acknowledgments}
The author would like to thank Xie Chen, Yu-An Chen, Aleksander G\l\'{o}dkowski, Jos\'{e} Garre-Rubio, Yingfei Gu, Ruochen Ma, Xiao-Qi Sun, Zongyuan Wang, and Yijian Zou for helpful discussions. This work was supported by the Chinese Academy of Sciences (CAS) under Grant No.\,YSBR-150 and a startup fund from the Institute of Physics, CAS. 

\appendix
\section{Gauging a Finite Group Symmetry}\label{app:GaugingReview}
In this section, we will briefly review the procedure for gauging an onsite finite group symmetry. A useful reference on this topic is Ref.\,\onlinecite{Verstraete2015Gauging}. 

\begin{figure}
    \centering
    \includegraphics[scale=\ScalingFactor]{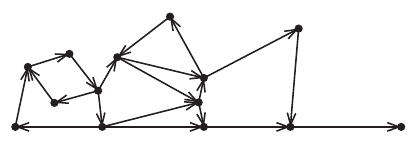}
    \caption{Example of a lattice consisting of vertices (sites) and oriented edges (links). }
    \label{fig:LatticeExample}
\end{figure}
Consider a lattice $\Lambda$ consisting of several vertices (sites) $v$ and edges (links) $e$, such as the example in Fig.\,\ref{fig:LatticeExample}. Our starting point is a quantum theory living on the vertices -- the theory of matter fields. More precisely, there is a local Hilbert space $\mc L_v$ on each vertex $v\in\Lambda$. such that the total matter field Hilbert space $\mc L_{\rm matter}=\bigotimes_{v\in\Lambda}\mc L_v$ is the tensor product of all those local ones. The quantum theory is described by a Hamiltonian $H$ acting on the Hilbert space $\mc L_{\rm matter}$. We assume that $H$ enjoys a finite symmetry group $G$. The symmetry action is $U_\Lambda (g)=\prod_{v\in \Lambda} U_v(g)$ where $g\in G$ and $U_v(\cdot)$ is the local unitary representation of $G$ on the vertex $v$. 

Our goal is to gauge the symmetry $G$, that is, promoting the global symmetry to local gauge symmetries. To this end, we need to introduce gauge field degrees of freedom living on the edges. The local Hilbert space on each edge is spanned by the orthonormal group element basis $\{\ket{g}|g\in G\}$. We choose an orientation for each edge and define a convenient rule: Reversing the
direction of a particular edge is equivalent to the basis change $\ket{g}\mapsto\ket{g^{-1}}$ on that edge. This rule will be useful for us to define new operators as we will only need to specify the operator action for a particular choice of the edge orientations. We will henceforth denote by $\mc L_e$ the local Hilbert space associated with the oriented edge $e$. The total gauge field Hilbert space is then given by $\mc L_{\rm gauge}=\bigotimes_{e\in\Lambda}\mc  L_e$. 
The starting and ending vertices of an edge $e$ will be denoted by $v_+(e)$ and $v_-(e)$, respectively. The set of all edges containing a vertex $v$ will be denoted by $E(v)$. 

Let $e$ be an arbitrary edge containing a vertex $v$ and let $h\in G$, we introduce an operator $X_{e,v}(h)$ defined as follows: 
\begin{align}
    X_{e,v}(h)\bigg|\includegraphics[valign=c,scale=\ScalingFactor]{InwardEdge.pdf}\bigg\rangle
    =\bigg|\includegraphics[valign=c,scale=\ScalingFactor]{InwardXAction.pdf}\bigg\rangle, 
\end{align}
where $g\in G$ labels the group element basis state. The case of the other edge orientation follows from the previous orientation reversing rule. More explicitly, $X_{e,v_-(e)}(h)\ket{g}_e=\ket{hg}_e$ and $X_{e,v_+(e)}(h)\ket{g}_e=\ket{gh^{-1}}_e$. Note that $X_{e,v_+(e)}(g)$ and $X_{e,v_-(e)}(h)$ commute. 
We can now introduce the local gauge symmetry generator $S_v(g)$ with the group element $g$ and around the vertex $v$ as
\begin{align}
    S_v(g):=U_v(g)\prod_{e\in E(v)}X_{e,v}(g)
\end{align}
which acts in $\mc L_{\rm matter}\otimes \mc L_{\rm gauge}$. 

The procedure for gauging the global symmetry $G$ consists of the following steps:
\begin{enumerate}
    \item Decompose the Hamiltonian $H$ into a sum of symmetric local operators (We assume there is a notion of locality and that $H$ is local). Apply a gauging map to each local operator $O$ so that the resulting operator $O'$ commutes with all gauge symmetry generators $S_v(g)$. 

    \item Optionally add flux terms into the gauged Hamiltonian obtained from the previous step if the lattice $\Lambda$ contains small loops. These new terms commute with all existing terms and with the gauge transformations $S_v(g)$. 

    \item Restrict the full Hilbert space $\mc L_{\rm matter}\otimes\mc L_{\rm gauge}$ to the gauge invariant subspace satisfying the Gauss's law constraints $S_v(g)=1$ for all $v\in\Lambda$ and $g\in G$. 
\end{enumerate}
In what follows, we will explicitly define the gauging map and the flux terms mentioned in the first and second steps, respectively. 

We will start by defining the flux terms. Let $f$ be an oriented loop along the edges and containing a vertex $v$ (see the picture in the equation below). The edges on the loop may or may not have consistent orientations with the loop. We define a projector $B_{f,v}(h)$ with $h\in G$ by the following equation. 
\begin{align}
    B_{f,v}(h)\Bigg|\includegraphics[valign=c,scale=\ScalingFactor]{LoopHolonomy.pdf}
    \Bigg\rangle
    =\delta_{g_1g_2g_3g_4,h}\Bigg|\includegraphics[valign=c,scale=\ScalingFactor]{LoopHolonomy.pdf}
    \Bigg\rangle. 
\end{align}
Here, we are showing a loop with four edges and all edge orientations are consistent with the loop orientation indicated by the circular arrow. The generalization to other cases should be clear given our edge orientation reversing rule stated previously. The product $g_1g_2g_3g_4$ in the above example is called the holonomy. $B_{f,v}(h)$ is the projector onto a particular value of the holonomy. One can check that $B_{f,v}(h)$ commutes with all $S_{v'}(g)$ with $v'\neq v$ and 
\begin{align}
    S_v(g)B_{f,v}(h)S_v(g)^\dagger=B_{f,v}(ghg^{-1}). 
\end{align}
Let $C\subset G$ be a conjugacy class, that is $C=\{ ghg^{-1}|g\in G \}$ for some fixed $h$. The following combination commutes with all $S_v(g)$: 
\begin{align}
    B_f(C):=\sum_{h\in C}B_{f,v}(h)
\end{align}
which projects onto a particular conjugacy class of the holonomy. In fact, $B_f(C)$ does not depend on the base point $v$ of the loop $f$ and hence we have omitted it from the notation. Since $B_f(C)$ is hermitian and commutes with all gauge symmetry generators, it can be added to the gauged Hamiltonian. Such terms are called flux terms, as the conjugacy class of holonomy is the analog of the exponentiated magnetic flux $e^{i\Phi}$ in U(1) gauge theories. In particular, the identity element $1\in G$ forms a conjugacy class by itself and it is a common choice to add $-uB_f(\{1\})$ with $u>0$ to the gauged Hamiltonian. These terms energetically favor trivial fluxes, in analogy with the magnetic Maxwell term in U(1) gauge theories. 

Let us now define the gauging map that acts on operators. Let $O$ be a  symmetric operator in the original matter field Hamiltonian. We would like to find a gauged operator $O'$ that satisfy the following properties. 
\begin{itemize}
    \item $O'$ should commute with all gauge symmetry generators $S_v(g)$. 

    \item In the subspace where all gauge field degrees of freedom are in the trivial state $\ket{1}$, $O'$ should reduce to $O$. 

    \item If $O$ is local, $O'$ should also be local. 
\end{itemize}
The first and second properties actually uniquely determine the action of $O'$ on flux free states. This is because any gauge field configuration with trivial flux in all loops can be mapped to the trivial configuration (all $\ket{1}$) by the action of $S_v(g)$ operators. However, the action of $O'$ on more general states are not fixed by the above requirements and the options are not unique. We will not attempt to classify all possibilities but instead give an explicit construction following Ref.\,\onlinecite{Verstraete2015Gauging}. 

The construction requires choosing a sublattice $\Gamma\subset \Lambda$ that contains the support of $O$. Furthermore, we require $\Gamma$ to satisfy two properties: (1) $\Gamma$ contains both vertices $v_{\pm}(e)$ of all its edges but not necessarily all edges of its vertices, and (2) $\Gamma$ is connected. 
Let $\{g_v|v\in\Gamma\}$ be a collection of group elements corresponding to every vertex in $\Gamma$, we define an operator $S_\Gamma(\{g_v\})$ as follows. 
\begin{align}
    S_\Gamma(\{g_v|v\in\Gamma\}):=\prod_{v\in\Gamma}U_v(g_v)\prod_{e\in E(v)\cap E(\Gamma)}X_{e,v}(g_v), 
\end{align}
where as mentioned previously, $E(v)$ is the set of all edges in $\Lambda$ containing the vertex $v$, and $E(\Gamma)$ is the set of all edges in $\Gamma$. The gauging map $\mc G_\Gamma$ acting on the operator $O$ is given by
\begin{align}
    O'&\equiv\mc G_\Gamma[O]\nonumber\\
    &=\frac{1}{|G|}\sum_{\{g_v|v\in\Gamma\}}S_\Gamma(\{g_v\})\left( O\bigotimes_{e\in\Gamma}\ket{1}_e\bra{1}_e \right)S_\Gamma(\{g_v\})^\dagger. 
    \label{eq:GaugingMapDef}
\end{align}
Recall that $O$ is an operator acting in $\mc L_{\rm matter}$ (sites) while after gauging, $O'$ acts in $\mc L_{\rm matter}\otimes \mc L_{\rm gauge}$ (both sites and edges). Implicitly in the equation above, $O'$ acts trivially on all edges not in $\Gamma$. Hence, if $\Gamma$ is local, so is $O'$. 
One can check that $O'$ has local gauge symmetries, i.e. it commutes with $S_v(g)$ for all $v\in\Lambda$ and $g\in G$. In fact, the construction is just to symmetrize the operator $O\bigotimes_{e\in\Gamma}\ket{1}_e\bra{1}_e$. With some algebra, one can also verify that 
\begin{align}
    \mc G_\Gamma[O]\left(\bigotimes_{e\in \Gamma}\ket{1}_e\right)=\left(\bigotimes_{e\in \Gamma}\ket{1}_e\right)\otimes O, 
\end{align}
where we have utilized the two conditions on the sublattice $\Gamma$ and the global $G$-invariance of $O$. Therefore, the gauged operator $O'$ given in \eqref{eq:GaugingMapDef} indeed meets our requirements. 

Let us mention a few properties of the gauging map $\mc G_\Gamma$ and the gauged operator $O'=\mc G_\Gamma[O]$. 
\begin{enumerate}
    \item Let $f\subset \Gamma$ be a loop contained in $\Gamma$. We have 
    \begin{align}
        B_f(\{1\})O'=O'B_f(\{1\})=O'. 
        \label{eq:GaugingMapProp1}
    \end{align}
    Hence, $O'$ contains projections onto trivial flux sectors for all loops in $\Gamma$. This explains why we should not choose $\Gamma$ to be the whole lattice $\Lambda$: In that way, $O'$ will generically be nonlocal. It also follows that $O'$ does depend on the choice of the sublattice $\Gamma$. 

    \item Let $O_1$ and $O_2$ be two symmetric operators both acting within $\Gamma$. One can check that 
    \begin{align}
        \mc G_\Gamma[O_1]\mc G_\Gamma[O_2]=\mc G_\Gamma[O_1O_2]. 
        \label{eq:GaugingMapProp2}
    \end{align}

    \item If $\Gamma$ is a tree-like sublattice, i.e. one that does not contain any loop, then 
    \begin{align}
        \mc G_\Gamma[1]=1. 
        \label{eq:GaugingMapProp3}
    \end{align}
    This and the previous property imply that if a set of operators $\{O_k\}$ forms a representation of a group $K$ and if $\Gamma$ is tree-like, then the set of gauged operators $\{O_k'\}$ also forms a representation of $K$. 

    \item If $\Gamma_1$ and $\Gamma_2$ are two tree-like sublattices such that $\Gamma_1\subset \Gamma_2$, then
    \begin{align}
        \mc G_{\Gamma_1}[O]=\mc G_{\Gamma_2}[O]
        \label{eq:GaugingMapProp4}
    \end{align}
    for any symmetric operator $O$ acting within $\Gamma_1$. 
\end{enumerate}

When the whole lattice $\Lambda$ is a $1d$ chain, the gauging map on local operators is unique and particularly nice. These are summarized as the following proposition and have been used in the main text. 
\begin{proposition*}
\textbf{(Gauging on a 1d Chain)}
    Suppose the whole lattice $\Lambda$ is a 1d chain or any other lattice without local loops. For any local operator $O$ acting on vertices (the matter field Hilbert space), there is a unique local operator $O'$ that commutes with all gauge symmetry generators $S_v(g)$ and reduces to $O$ on the trivial gauge field configuration $\bigotimes_e\ket{1}_e$. Moreover, this gauging map $(\cdot)'$ satisfies $O_1'O_2'=(O_1O_2)'$ and $1'=1$. 
\end{proposition*}
\begin{proof}
    Let $O'_{\rm I}$ and $O'_{\rm II}$ be two local operators that both commute with all $S_v(g)$ and both reduce to $O$ on the trivial gauge field configuration. We need to show that $O'_{\rm I}-O'_{\rm II}=0$. Consider a tensor product state of the form $\bigotimes_v\ket{\psi_v}_v\bigotimes_e\ket{g_e}_e$. Let $\Gamma$ be a local sublattice containing the supports of both $O'_{\rm I}$ and $O'_{\rm II}$. $\Gamma$ has no loop by assumption and therefore, we can use a sequence of $S_v(g)$ operators to map all gauge field states within $\Gamma$ to the trivial state $\ket{1}$. It follows that $O'_{\rm I}-O'_{\rm II}$ annihilates all such tensor product states and thus vanishes identically. 

    The relations $O_1'O_2'=(O_1O_2)'$ and $1'=1$ follow from Eqs.\,\ref{eq:GaugingMapProp2} and \ref{eq:GaugingMapProp3} which may be verfied by direct computations. 
\end{proof}

\section{Minimal Example with an SPT initial Hamiltonian}\label{App:SPTExample}
In this appendix, we curiously apply our layered gauging method to decoupled layers of $1d$ SPTs. 

We start from a $2d$ square lattice as in Fig.\,\ref{fig:SPTExampleLattice}a. The initial decoupled Hamiltonian is 
\begin{align}
    H_0= -\sum_{n,j}Z_{n,j-1}X_{n,j}Z_{n,j+1}. 
\end{align}
Consider a $\mathbb{Z}_2\times \mathbb{Z}_2$ symmetry acting on every layer $n$, generated by $U_{n,o}=\prod_{k}X_{n,2k-1}$ and $U_{n,e}=\prod_{k}X_{n,2k}$. Each layer of the decoupled Hamiltonian $H_0$ realizes an SPT phase under the corresponding $\mathbb{Z}_2^2$ symmetry. 

\begin{figure}
    \centering
    \includegraphics[scale=\ScalingFactor]{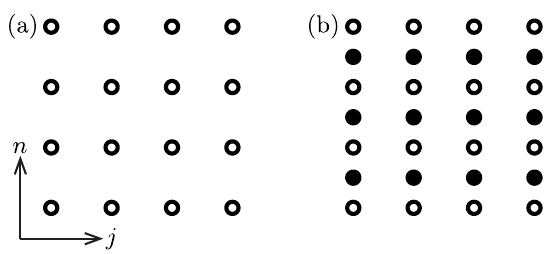}
    \caption{(a) Initial and (b) final bulk lattices for the SPT example. }
    \label{fig:SPTExampleLattice}
\end{figure}

After applying the layered gauging, the bulk lattice becomes that in Fig.\,\ref{fig:SPTExampleLattice}b. Note that the black dots with odd (even) $j$ are from the gauging of $U_{n,e}$ ($U_{n,o}$), and the Gauss's law constraints read
\begin{align}
    \includegraphics[valign=c,scale=\ScalingFactor]{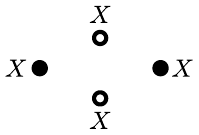}=1. 
\end{align}
The bulk Hamiltonian $\tilde H$ consists the following five-body stabilizers (which enter $\tilde H$ with a $-1$ coefficient): 
\begin{align}
    \includegraphics[valign=c,scale=\ScalingFactor]{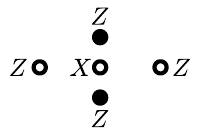}. 
\end{align}
This Hamiltonian actually realizes a $\mathbb{Z}_2\times \mathbb{Z}_2$ topological order, i.e. equivalent to two copies of toric code. To see this, apply $\CZ$ gates between every pair of (1) {vertically neighboring sites} and (2) horizontally neighboring \emph{white} sites. After this transformation, the stabilizers in the Hamiltonian become single $X$'s on white sites. The white sites are therefore pinned to the $\ket{+}$ state and can be effectively removed. In the remaining square lattice of black sites, the Gauss's law constraints read
\begin{align}
    \includegraphics[valign=c,scale=\ScalingFactor]{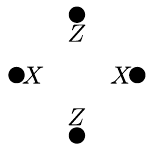}. 
\end{align}
We see that the even and odd sublattices of the black sites each supports a ground state of the Wen plaquette model \cite{WenPlaquette} which is equivalent to the toric code. 

\section{Explicit Form of the Semion Loop Operator}\label{app:DS}
In Section \ref{sec:AnomalyExample}, we mentioned that the product of all $S^{ER}$ stabilizers in a 2$d$ region equals to a loop operator supported on the region boundary. An example of such a loop operator for a rectangular region is shown in Fig.\,\ref{fig:ERLoop}. 
\begin{figure*}
    \centering
    \includegraphics[scale=\ScalingFactor]{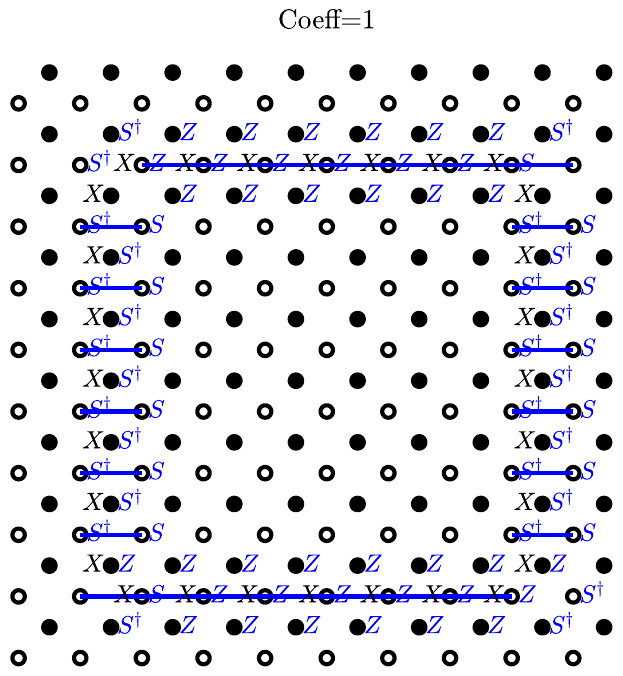}
    \caption{A $ER$-loop operator from the product of $S^{ER}$ stabilizers in a rectangular region. The overall coefficient, which is not important for our purpose, depends on the rectangular size and equals to $+1$ for the example shown. }
    \label{fig:ERLoop}
\end{figure*}

\bibliography{Bib_LayeredGauging}
\end{document}